\documentclass{article}
\usepackage{graphicx} 
\usepackage{amsmath}
\usepackage{amsthm}
\usepackage{amssymb}
\usepackage{fullpage}
\usepackage{complexity}
\usepackage[hidelinks, colorlinks=true, allcolors=blue!70!black]{hyperref}

\newtheorem{theorem}{Theorem}[section]
\newtheorem{lemma}{Lemma}[section]
\newtheorem{definition}{Definition}[section]
\newtheorem{openproblem}{Open Problem}[section]
\newtheorem{corollary}{Corollary}[section]
\newtheorem{remark}{Remark}[section]

\renewcommand{\epsilon}{\varepsilon}
\newcommand{\Z}{\mathbb{Z}}
\newcommand{\eps}{\varepsilon}
\newcommand{\N}{\mathbb{N}}
\newcommand{\F}{\mathbb{F}}
\renewcommand{\R}{\mathbb{R}}

\DeclareMathOperator{\size}{size}
\DeclareMathOperator{\asize}{asize}

\newlang{\SUM}{SUM}
\newlang{\APSP}{APSP}
\newlang{\TAUT}{TAUT}
\newlang{\MAX}{MAX}

\newclass{\coNTIME}{coNTIME}
\newclass{\POTIME}{POTIME}
\newclass{\PONTIME}{PONTIME}
\newclass{\LS}{LS}

\usepackage{tikz}
\usetikzlibrary{arrows.meta, graphs}

\title{Computations with polynomial evaluation oracle:\\
ruling out superlinear SETH-based lower bounds}
\author{
Tatiana Belova
\thanks{Steklov Institute of Mathematics at St. Petersburg. Email: \texttt{yukikomodo@gmail.com}.}
\and
Alexander~S. Kulikov
\thanks{JetBrains Research. Email: \texttt{alexander.s.kulikov@gmail.com}.}
\and
Ivan Mihajlin
\thanks{JetBrains Research. Email: \texttt{ivmihajlin@gmail.com}.}
\and
Olga Ratseeva
\thanks{St.~Petersburg State University. Email: \texttt{ratsionalno@mail.ru}.}
\and
Grigory Reznikov
\thanks{Steklov Institute of Mathematics at St. Petersburg. Email: \texttt{gritukan@gmail.com}.}
\and
Denil Sharipov 
\thanks{St.~Petersburg State University. Email: \texttt{denil.sharipov@rambler.ru}.}
}
\date{}

\begin{document}

\sloppy
\maketitle

\begin{abstract}
    The field of~fine-grained complexity aims 
    at~proving conditional lower bounds 
    on~the time complexity of~computational problems.
    One of~the most popular and successfully used assumptions, Strong Exponential Time Hypothesis (SETH), implies that \SAT{} cannot be~solved in~$2^{(1-\epsilon)n}$ time. In~recent years, 
    it~has been proved that known algorithms for many problems are
    optimal under SETH. 
    Despite wide applicability of~SETH, 
    for many problems,
    there are no~known 
    SETH-based lower bounds,
    so~the quest for new reductions continues.
    
    Two barriers for proving SETH-based lower bounds are known.
    Carmosino et al. (ITCS 2016) introduced the Nondeterministic Strong Exponential Time Hypothesis 
    (NSETH) stating that
    \TAUT{} cannot be~solved in~time $2^{(1-\epsilon)n}$
    even if~one allows
    nondeterminism. They used this hypothesis to~show 
    that some natural fine-grained reductions would 
    be~difficult to~obtain: proving that, say,
    $3$-\SUM{} 
    requires time $n^{1.5+\epsilon}$ under SETH, breaks NSETH and this, 
    in~turn, implies strong circuit lower bounds.
    Recently, Belova et~al. (SODA 2023) introduced the so-called polynomial formulations to~show that for many
    \NP{}-hard problems, proving any explicit exponential lower bound under SETH also implies strong circuit lower bounds.
    
    In~this paper, we~combine the two barriers above.
    We~prove that for a~range of~problems from~\P{}, including $k$-\SUM{} and triangle detection,
    proving superlinear lower bounds under SETH is~challenging
    as~it~implies new circuit lower bounds.
    To~this end, we~show that these problems
    can be solved in~nearly linear time
    with oracle calls to~evaluating a~polynomial of~constant degree.
    Then, we~introduce a~strengthening of~SETH stating that solving \SAT{} in time $2^{(1-\varepsilon)n}$ is~difficult even if~one has constant degree polynomial evaluation oracle calls. This hypothesis is~stronger and less
    believable than SETH, but refuting~it is~still challenging:
    we~show that this implies circuit lower bounds.
    
    Finally, by~considering computations that make oracle calls to~evaluating constant degree polynomials depending on~a~small number
    of~variables,
    we~show connections between nondeterministic 
    time lower bounds and arithmetic circuit lower bounds.
    Namely, we~prove that if~any
    of~\MAX{}-$k$-\SAT{},
    Binary Permanent,
    {or a~variant of~Set Cover}
    problems cannot 
    be~solved in~co-nondeterministic time $2^{(1-\varepsilon)n}$, for any $\varepsilon>0$,
    then one gets 
    arbitrary large polynomial arithmetic circuit lower bounds.
    \end{abstract}

\clearpage
\tableofcontents
\clearpage



\section{Overview}
\subsection{Background}
\subsubsection{Fine-grained complexity: conditional hardness of~algorithmic problems}
The field of~fine-grained complexity aims 
at~proving conditional lower bounds 
on~the time complexity of~computational problems. 
One of~the most popular and successfully used assumptions, Strong Exponential Time Hypothesis (SETH)~\cite{DBLP:journals/jcss/ImpagliazzoPZ01}, states that, for any $\varepsilon>0$, there exists~$k$ such that $k$-\SAT{} cannot be~solved in~$2^{(1-\epsilon)n}$ time. 
SETH has been used with great success to~prove conditional hardness results for a~variety of~problems from a~wide range of~domains: 
$n^2$ for Orthogonal Vectors~\cite{DBLP:journals/tcs/Williams05} (where $n$~is the number of~vectors), $3/2$-approximate Graph Diameter~\cite{DBLP:conf/stoc/RodittyW13} (where $n$~is the number of nodes in the input graph), and {Edit Distance}~\cite{DBLP:journals/siamcomp/BackursI18} (where $n$~is the length of~the input strings);
$2^{n}$ for Hitting~Set (where $n$~is the size of the universe) and NAE-SAT (where $n$~is the number of the variables)~\cite{DBLP:journals/talg/CyganDLMNOPSW16};
$n^k$ for $k$-Dominating Set~\cite{DBLP:conf/soda/PatrascuW10} (where $n$~is the number of~nodes in~the input graph and $k\geq7$);
$2^{\operatorname{tw}}$ for Independent Set~\cite{DBLP:journals/talg/LokshtanovMS18} (where $\operatorname{tw}$ is the treewidth of the input graph).

Despite wide applicability of~SETH, 
for many problems,
there are no~known 
SETH-based lower bounds,
so~the quest for new reductions continues. For example,
even assuming SETH, it~is not excluded that
\begin{itemize}
    \item $k$-\SUM{} can be~solved in~time $n^{1.1}$ (where $n$~is the number of~input integers);
    \item Collinearity (three points on a~line) can be~solved 
    in~time $n^{1.1}$ (where $n$~is the number of~input points);
    \item 4-Cycle Detection can be~solved in~time $(n+m)^{1.1}$ (for a~graph with $n$~nodes and $m$~edges).
\end{itemize}
In~this paper, we~show that excluding any of~these possibilities under SETH is~difficult as~it~gives
new circuit lower bounds.

\subsubsection{Hardness of~showing hardness: known barriers}
In~2016, 
Carmosino et~al.~\cite{DBLP:conf/innovations/CarmosinoGIMPS16}
gave an~explanation for the lack of~SETH-based lower bounds
for some problems (that can be~solved efficiently
both nondeterministically and co-nondeterministically).
They introduced the Nondeterministic Strong Exponential Time Hypothesis (NSETH) stating that 
for any $\epsilon>0$, there exists~$k$ such that
$k$-\TAUT{} cannot be~solved in \emph{nondeterministic}
time $2^{(1-\epsilon)n}$ (in~the $k$-\TAUT{} problem, one is given a~$k$-CNF formula and is~asked to~check whether it is unsatisfiable).
They used this hypothesis to~show 
that some natural fine-grained reductions would 
be~difficult to~obtain: proving that, say,
$3$-\SUM{} requires time $n^{1.51}$ under SETH breaks NSETH and this, 
in~turn, implies strong circuit lower bounds.
Since then, NSETH has been used to~show hardness of~proving SETH hardness for various problems: Viterbi 
Path~\cite{DBLP:conf/icml/BackursT17}, 
Approximate Diameter~\cite{DBLP:conf/stoc/Li21a},
All-Pairs Max-Flow~\cite{DBLP:journals/corr/abs-2304-04667}.
Also, 
Aggarwal and Kumar~\cite{DBLP:journals/corr/abs-2211-04385}
showed hardness of~showing hardness for variants of~the Closest Vector Problem.

Recently, Belova et~al.~\cite{DBLP:conf/soda/BelovaGKMS23} introduced the so-called polynomial formulations.
Informally, a~polynomial formulation of~size~$s(n)$ for 
a~computational problem~$A$ is~a~family $\mathcal{P}=\{P_n\}_{n \in \mathbb N}$ of~constant degree polynomials such that to~solve any instance of~$A$ of~size~$n$ it suffices to~encode~it
as~a~sequence of~integers and to~evaluate the corresponding polynomial $P_{s(n)}$.
They proved that various \NP{}-hard problems admit 
a~polynomial formulation of~size $\alpha^n$ for any $\alpha>1$
(the degree of~the corresponding polynomial depends on~$\alpha$)
and used this to~show that proving an~explicit exponential lower bound for any of~these problems under SETH implies new circuit lower bounds.

\subsection{New Results}
\subsubsection{An~even stronger version of~SETH that is~still hard to~refute}
A~polynomial formulation may be~viewed as a~many-one reduction
to~polynomial evaluation.
In~this paper, we~generalize this to~Turing reductions
and 
consider computations that make oracle calls
to~polynomial evaluation. 

As~proved in~\cite{DBLP:journals/iandc/JahanjouMV18},
refuting SETH implies Boolean circuit lower bounds. 
In~Section~\ref{sec:poseth},
we~introduce the following stronger version of~SETH.
\begin{quote}
    \textbf{POSETH:} for every~$\epsilon>0$ and every explicit family of~polynomials~$\mathcal Q$, there exists~$k$ such that $k$-\TAUT{} cannot be~solved in time $O(2^{(1-\epsilon)n})$ 
    by~a~deterministic algorithm with oracle calls 
    to~evaluating~$\mathcal Q$.
\end{quote}

We~prove that, perhaps surprisingly, refuting it~still implies interesting circuit lower bounds: if~POSETH fails,
then either $\E^{\NP}$ cannot be~computed by~Boolean series-parallel circuits of~linear size or~there exists an~explicit 
family of~polynomials 
of~constant degree 
whose arithmetic circuit complexity is~truly superlinear. 
Proving any of~these two
circuit lower bounds remains
a~major open problem for the last four decades.
We~also state a~nondeterministic version (PONSETH)
of~the new hypothesis and show that refuting~it
gives the same lower bounds.

The picture below summarizes the discussed
hypotheses together with consequences of~refuting them.
There, $\size(\cdot)$~and~$\size_{sp}(\cdot)$ denote circuit size and series-parallel circuit size, whereas $\asize(\cdot)$~denotes arithmetic circuit size. We~define formally these complexity measures in~Section~\ref{sec:circuits}.

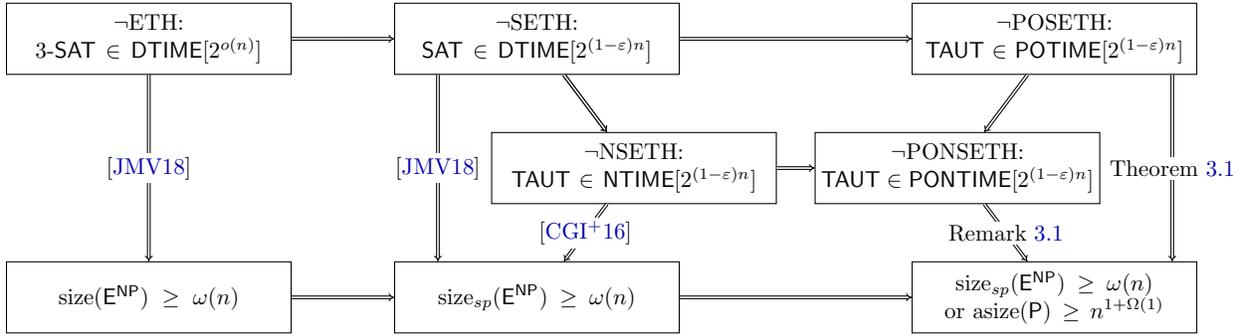
\begin{center}
    \begin{tikzpicture}[scale=.86, transform shape]
        \tikzstyle{l} = [rectangle, draw, inner sep=.5mm, text width=43mm, align=center, minimum height=11mm]
        \tikzstyle{i} = [double, -Implies]
        
        \node[l] (e) at (0,4) {
            $\neg${ETH:}\\ 
            $\text{$3$-\SAT} \in \DTIME[2^{o(n)}]$
        };
        
        \node[l] (s) at (6,4) {
            $\neg${SETH:}\\ 
            $\SAT \in \DTIME[2^{(1-\epsilon)n}]$
        };
        
        \node[l] (n) at (7.5,2) {
            $\neg${NSETH:}\\ 
            $\TAUT \in \NTIME[2^{(1-\epsilon)n}]$
        };
        
        \node[l] (p) at (14,4) {
            $\neg${POSETH:}\\ 
            $\TAUT \in \POTIME[2^{(1-\epsilon)n}]$
        };

        \node[l] (pn) at (12.5,2) {
            $\neg${PONSETH:}\\ 
            $\TAUT \in \PONTIME[2^{(1-\epsilon)n}]$
        };
        
        \foreach \f/\t in {e/s, s/p, s/n, n/pn, p/pn}
            \draw[i] (\f) -- (\t);

        \draw[i] (s) -- (n);
        
        \node[l] (ea) at (0,0) {$\size(\E^\NP) \ge \omega(n)$};
        \node[l] (sa) at (6,0) {$\size_{sp}(\E^\NP) \ge \omega(n)$};
        \node[l] (pa) at (14,0) {$\size_{sp}(\E^\NP) \ge \omega(n)$ or $\asize(\P) \ge n^{1+\Omega(1)}$};

        \foreach \f/\t in {ea/sa, sa/pa}
            \draw[i] (\f) -- (\t);
        
        \tikzstyle{qq} = [midway, rectangle, fill=white, inner sep=.5mm, align=center]
        
        \draw[i] (e) -- (ea) node[qq] {\cite{DBLP:journals/iandc/JahanjouMV18}};
        \draw[i] (s.-160) -- (sa.160) node[qq] {\cite{DBLP:journals/iandc/JahanjouMV18}};
        \draw[i] (n) -- (sa) node[qq] {\cite{DBLP:conf/innovations/CarmosinoGIMPS16}};
        \draw[i] (p.-17) -- (pa.17) node[qq] {Theorem~\ref{thm:pnsethlowerbounds}};
        \draw[i] (pn) -- (pa) node[qq] {Remark~\ref{rem:pn}};
    \end{tikzpicture}
\end{center}

\subsubsection{Hardness of~proving SETH hardness for various problems from~\P{}}
In~Section~\ref{sec:polyp}, we~show that various
problems from~\P{} admit 
a~polynomial formulation of~size $O(n^{1+\varepsilon})$,
for any $\varepsilon>0$ (the degree of~the corresponding polynomial family depends on~$\varepsilon$).
We~do this for all problems
where the goal is~to~find a~subset of~constant size
that can be~verified ``locally'' (we~make it~formal later
in~the text). This includes problems like: $k$-\SUM{}, collinearity (three points on a~line), minimum weight $k$-clique, MAX $H$-SUBGRAPH, $\mathcal{H}$-induced subgraph problem where $\mathcal H$~is a~finite set of~graphs. Finally, we~show
that proving a~truly superlinear (that~is, of~the form $n^{1+\epsilon}$ for $\epsilon>0$) SETH-based lower bound
for any of~these problems breaks POSETH and hence implies 
circuit lower bounds.
This way, we~do~not only obtain new hardness-of-showing-hardness results, but also subsume some of~the existing ones.
For example, Carmosino et~al.~\cite{DBLP:conf/innovations/CarmosinoGIMPS16} show that it~is challenging 
to~prove an~$n^{1.5+\epsilon}$ SETH-based lower bound for $3$-\SUM{}, whereas we~show that it~is in~fact difficult to~get even an~$n^{1+\epsilon}$ SETH-based lower bound and even for
$k$-\SUM{} for any~$k$.

It~is interesting to~compare our hardness-of-hardness results
with known lower bounds. For example, it~is known that the Edit Distance problem is $n^2$-hard~\cite{DBLP:journals/siamcomp/BackursI18} 
where the $k$-Dominating Set problem is~$n^k$-hard~\cite{DBLP:conf/soda/PatrascuW10} under SETH.
Both these problems are indeed ``non-local'':
it~is not clear how to~state any of~them as~a~problem of~finding
a~constant-size subset of~the input that can be~verified
locally (without looking at~the rest of~the input).

\subsubsection{Connections between nondeterministic complexity and arithmetic circuit complexity}

Recall that in~\cite{DBLP:conf/soda/BelovaGKMS23}, it~is proved that 
for any $\alpha>1$, there exists $\Delta=\Delta(\alpha)$ such that
$k$-\SAT{}, \MAX{}-$k$-\SAT{}, Hamiltonian Path, Graph Coloring, Set Cover,
Independent Set, Clique, Vertex Cover, and $3d$-Matching problems
admit a~$\Delta$-polynomial formulation of~complexity~$\alpha^n$.
That~is, at~the expense of~increasing the degree of a~polynomial,
one can make the number of~variables in~this polynomial 
an~arbitrary small exponential function. In~Section~\ref{sec:polyar}, we~generalize this result and consider computations with polynomial evaluation oracle whose all oracle calls compute a~polynomial
that is not too~large.
We~show that \MAX{}-$k$-\SAT{}, 
{Set Cover with $n^{O(1)}$ sets of size at most $O(\log{n})$}, and Binary Permanent problems
can be~solved in~deterministic time $2^{(1-\varepsilon)n}$
by an~algorithm that makes oracle calls to~evaluating a~polynomial
whose number of~variables is~not too large. Using this, we~show that if~any of~these problem is not in~$\coNTIME[O(2^{(1-\delta)n})]$
for any~$\delta>0$, then, for any~$\gamma$, there exists an~explicit polynomial family that cannot be~computed by~arithmetic circuits of~size $O(n^\gamma)$.
It~is interesting to~compare this result with 
a~result by~\cite{DBLP:journals/iandc/JahanjouMV18}
that says that refuting SETH implies Boolean circuit lower bounds:
both results derive circuit lower bounds, but 
\cite{DBLP:journals/iandc/JahanjouMV18} derives it~from
time upper bounds, whereas we~derive~it
from time lower bounds.

\section{General setting}
\subsection{Computational problems and computational model}
We~study the following computational problems. 

\begin{itemize}
    \item Problems from~\P{}. In~all of~them,
    the parameter~$W$ grows polynomially in~$n$
    (that is, $W=n^{O(1)}$) and $R=\{-W,-W+1, \dotsc, W\}$
    denotes the range of~possible values of~the input integers.
    \begin{itemize}
        \item $k$-\SUM{}: given $k$~sets $S_1, \dotsc, S_k \subseteq R$ of~size~$n$, check whether it~is possible to~select a~single integer from each set such that the sum of~the $k$~selected integers is~zero.
        \item Collinearity (three points on a~line): given 
        $n$~points in $R^2$, check whether three of~them
        lie 
        on a~single line.
        \item $H$-induced subgraph: given a~graph with $n$~vertices and $m$~edges, check whether it~contains an~induced 
        subgraph isomorphic to $H$.  
        \item $\mathcal{H}$-induced subgraph: given a~graph with 
        $n$~vertices and $m$~edges, check whether it~contains 
        an~induced subgraph isomorphic to one of the graphs from the finite set~$\mathcal{H}$.  
    
        \item Minimum weight $k$-clique (decision version): given 
        a~graph with $n$~vertices and $m$~edges
        with weights from~$R$
        as~well as~an~integer~$w$,
        check whether 
        it~contains a~$k$-clique of~weight at~most~$w$.
        
        \item Maximum $H$-subgraph (decision version): given a~graph
        with $n$~vertices and $m$~edges with weights on vertices (or edges) from~$R$ as~well~as an~integer~$w$, check whether it~contains an induced subgraph isomorphic to~$H$ of~total weight at least~$w$.
    \end{itemize}
    \item \NP{}-hard problems.
    \begin{itemize}
        \item $k$-\SAT{}: given a~$k$-CNF formula over $n$~Boolean variables, check whether 
        it~is satisfiable.
        \item $k$-\TAUT{}: given a~$k$-CNF formula over $n$~Boolean variables, check whether 
        it~is unsatisfiable.
        \item \MAX{}-$k$-\SAT{}: given a~$k$-CNF formula over $n$~Boolean variables, find the maximum number of~clauses that 
        can be~satisfied simultaneously.
        \item Binary Permanent: given an $n \times n$ matrix $A = (a_{i,j})$ over $\{0, 1\}$, compute \[\operatorname{perm}(A) = \sum\limits_{\sigma \in S_n} \prod\limits_{i \in [n]}a_{i, \sigma(i)}\,,\] where $S_n$ is a set of all permutations of the numbers $1, 2, \dots, n$.
        \item Set Cover: given a set family $\mathcal{F} \subseteq 2^{[n]}$ of size $n^{O(1)}$, find the minimum number of sets from $\mathcal{F}$ that cover $[n]$.
    \end{itemize}
\end{itemize}



For a~decision problem~$A$ with a~complexity parameter~$n$,
by~$A_n$ we~denote
the set of~instances of~$A$ of~size~$n$,
whereas $A'_n \subseteq A$ is a~set of yes-instances.
Thus, an~algorithm for solving~$A$ must be~able to~check, for any~$n$ and any instance $I \in A_n$, whether $I \in A'_n$.
In~all problems considered in~this paper, 
the bit-size of~any instance $I \in A_n$ is~bounded 
by a~polynomial in~$n$.




We~assume that an~algorithm is~performed on a~RAM machine that on~an~input of~size~$n$, can perform standard 
arithmetic operations with $O(\log n)$-bit integers 
in~unit time.

\subsection{Explicit polynomial families}
By a~\emph{family of~polynomials~$\mathcal P$}
we~mean an~infinite sequence
of~polynomials $\{P_n\}_{n \in \mathbb N}$
such that $P_n$ is a~multivariate polynomial
depending on $n$~variables.
We say that $\mathcal P$~has \emph{degree~$d(n)$}
if, for every~$n$, every monomial of~$P_n$ has total degree
at~most~$d(n)$.
In~this paper, we~consider polynomials 
over the rings $\mathbb{Z}$ and $\mathbb{Z}_p$ where $p$ is a~prime number.

\begin{definition}[explicit family of~polynomials]
    \label{def:explicitfamily}
    We~say that a~family~$\mathcal P$ of~polynomials
    is~\emph{explicit}, if there exists a~constant~$\Delta$ such that for all~$n$,
    the degree of~$P_n$ is~at~most $\Delta$, 
    {the absolute value of~every coefficient 
    of~$P_n$ is~at~most $O(n^{\Delta})$,} and
    all of~them can be~computed (simultaneously) in~time $O(n^{\Delta})$.
\end{definition}

If $\mathcal P$ is~explicit, then the number
of~monomials of~$P_n$ is~at~most $\sum_{i=0}^{\Delta}n^i \le (\Delta+1)n^{\Delta}=O(n^{\Delta})$. Thus, 
for any $x_1, \dotsc, x_n$ 
such that $|x_i| \le \rho$ for all $i \in [n]$,
\begin{equation}\label{eq:explicitpolyupper}
    |P_n(x_1, \dotsc, x_n)| \le O(n^{\Delta} \cdot n^{\Delta} \cdot \rho^\Delta) =O((n\rho)^{2\Delta}) \, .
\end{equation}
\color{black}

\subsection{Boolean and arithmetic circuits}
\label{sec:circuits}
\begin{definition}
A~Boolean circuit~$C$ with variables $x_1,\dotsc,x_n$ is a~directed acyclic graph whose every node has in-degree zero or two. The in-degree zero nodes are labeled either by variables $x_i$ or constants~$0$ or~$1$. The in-degree two nodes are labeled by binary Boolean functions that map $\{0,1\}^2$ to $\{0,1\}$. The only gate of out-degree zero is the output of the circuit.
\end{definition}
A Boolean circuit $C$ with variables $x_1,\ldots,x_n$ computes a Boolean function $f\colon\{0,1\}^n\to\{0,1\}$ in a natural way. We define the size of~$C$ (denoted $\size(C)$) as the number of gates in~it, and the Boolean circuit complexity of a function as the minimum size of a circuit computing it.

A circuit is called \emph{series-parallel} if there exists a numbering $\ell$ of the circuit's nodes such that for every wire $(u,v), \ell(u)<\ell(v)$, and no pair of wires $(u,v), (u',v')$ satisfies $\ell(u)<\ell(u')<\ell(v)<\ell(v')$.

The best known lower bound on the size of a Boolean circuit for functions in~\P{} is~${3.1n-o(n)}$~\cite{DBLP:conf/stoc/Li022}. In fact, this bound remains the best known bound even for the much larger class of languages $\E^{\NP}$ even against the~restricted model of~series-parallel circuits. A~long-standing open problem in Boolean circuit complexity is to find an explicit language that cannot be computed by linear-size circuits from various restricted circuit classes (see~\cite{DBLP:conf/mfcs/Valiant77} and \cite[Frontier~3]{DBLP:books/daglib/0023084}).

\begin{openproblem}
Prove a lower bound of $\omega(n)$ on the size of Boolean series-parallel circuits computing a language from $\E^{\NP}$.
\end{openproblem}

\begin{definition}
An arithmetic circuit~$C$ over a ring~$R$ and variables $x_1,\dotsc,x_n$ is a directed acyclic graph whose every node has in-degree zero or two. The in-degree zero nodes are labeled either by variables $x_i$ or elements of $R$. The in-degree two nodes are labeled by either $+$ or $\times$. Every gate of out-degree zero is called an output gate.
\end{definition}
A single-output arithmetic circuit~$C$ over~$R$ computes 
a~polynomial over $R$. We say that $C$ computes a polynomial $P(x_1,\dotsc,x_n)$, if the two polynomials are \emph{identical} (as opposed to saying that $C$ computes $P$ if the two polynomials agree on every assignments of $(x_1,\ldots,x_n)\in R^n$). We define the size of~$C$ (denoted $\asize(C)$) as the number of edges in~it, and the arithmetic circuit complexity of a polynomial as the minimum size of a circuit computing it.

While it~is known~\cite{s73,bs83} that the polynomial $x_1^r+\ldots+x_n^r$ requires arithmetic circuits over $\F$ of size $\Omega(n\log r)$ (if $r$ does not divide the characteristic of~$\F$), one of the biggest challenges in algebraic complexity is to prove stronger lower bounds on the arithmetic circuit complexity of an~explicit family of~polynomials.

\begin{openproblem}
\label{op:ackts}
Find an~explicit polynomial family that requires arithmetic circuits of~size $\Omega(n^{\gamma})$ for
a~constant $\gamma>1$.
\end{openproblem}

\subsection{Minimum arithmetic circuit problem}\label{sec:macp}
Below, we~consider the problem of~finding an~arithmetic circuit
of~minimum size computing the given polynomial. It turns out
that if~the polynomial has constant degree, then one can guess and verify an~arithmetic
circuit of~size close to~optimal in~polynomial time.

\begin{definition}
    Let $p$~be a~prime number, $\mathcal P=\{P_n\}_{n \in \mathbb N}$ be a~family of polynomials over $\Z_p$, and $c \ge 1$ be an~integer parameter.
    The problem $\operatorname{Gap-MACP}_{\mathcal P, c, p}(n, s)$ is: given $n, s \in \mathbb N$, output an~arithmetic circuit over $\Z_p$ of~size at most~$cs$ computing~$P_n$, if $P_n$ can be~computed in~$\Z_p$ by a~circuit of size at~most~$s$;
    otherwise, output anything.
\end{definition}

\begin{lemma}[Lemma~3.5 in~\cite{DBLP:conf/soda/BelovaGKMS23}]
    \label{lemma:macp}
    There exists a~constant $\mu>0$ such that for every $\Delta$-explicit family of~polynomials~$\mathcal P$ and every prime number~$p$, 
    \[\operatorname{Gap-MACP}_{\mathcal P, \mu \Delta^2, p}(n, s) \in \NTIME[O(sn^{2\Delta}\log^2p)] \, .\]
\end{lemma}

The proof of~the lemma proceeds as~follows.
As~proved by~Strassen~\cite{strassen1973vermeidung},
for a~polynomial~$P$ of~degree~$\Delta$
and arithmetic circuit size~$s$,
there exists an arithmetic circuit
of~size $O(\Delta^2s)=O(s)$ all of~whose gates compute polynomials of~degree at~most~$\Delta$.
As~proved by~Strassen~\cite{strassen1973vermeidung},
for a~polynomial~$P$ of~degree~$\Delta$
and arithmetic circuit size~$s$,
there exists a~\emph{homogeneous} arithmetic circuit
of~size $O(\Delta^2s)=O(s)$: 
all its gates compute polynomials of~degree at~most~$\Delta$.
Verifying a~homogeneous circuit of~constant degree 
is~particularly easy: one just expands every gate
as a~polynomial and compares (monomial by~monomial) 
the resulting polynomial to~$P$.
Below, we~show that one can quickly
guess and verify an~arithmetic circuit 
computing a~polynomial from an~explicit family.

\begin{lemma}
    \label{lemma:macptwo}
    Let $\mathcal P=\{P_n\}_{n \in \mathbb N}$ be an~explicit family of~polynomials
    and $\rho$ be a~positive integer.
    For any~$n$, if $P_n$~has arithmetic circuit size~$s$,
    one can find an~arithmetic circuit $C_n$ of~size $O(s)$
    {and a~prime number $p=O((n\rho)^{2\Delta})$}, such that
    \[{|P_n(x)|<p} \text{ and }  P_n(x)=C_n(x) \text{ for all } x=(x_1, \dotsc, x_n) \in \mathbb{Z}^n \text{ satisfying } |x_i| \le \rho \text{ for all $i \in [n]$},\]
    in~nondeterministic time
    \[O\left(sn^{2\Delta}\log^2(n\rho) \right) \, .\]
\end{lemma}
\begin{proof}
    By~\eqref{eq:explicitpolyupper}, $|P_n(x_1, \dotsc, x_n)| < M$, where $M=O((n\rho)^{2\Delta})$,
    for any $(x_1, \dotsc, x_n)$ such that $|x_i| \le \rho$ for all $i \in [n]$.
    Hence, it~suffices to~compute $P_n$ modulo a~prime $2M \le p \le 4M$. One can generate~$p$
    in~nondeterministic time $O(\log^7M)=O((\log((n\rho)^{2\Delta}))^7)=O(\log^7(n\rho))$~\cite{aks04,lp19}.
    Let $P'_n$ be~$P_n$ with all coefficients taken modulo~$p$
    and $\mathcal{P}'=\{P_n'\}_{n \in \mathbb N}$.
    Then, $\asize(P'_n)=O(s)$ and a~call to
    \[\operatorname{Gap-MACP}_{\mathcal P', \mu \Delta^2, p}(n, s)\]
    (see Lemma~\ref{lemma:macp})
    returns an~arithmetic circuit~$C_n$ of~size $O(s)$
    computing $P'_n$ in $\mathbb{Z}_{p}$ in~nondeterministic time
    {
    \[ O(sn^{2\Delta}\log^2p)=O(sn^{2\Delta}\log^2(n\rho))\, .\]}
\end{proof}
\begin{remark}\label{rem:rho}
    If {$\rho=2^{n^{o(1)}}$},   
    the upper 
    bound in~Lemma~\ref{lemma:macptwo} can~be
    stated in a~more convenient form:
    \[O\left(sn^{2\Delta+1}\right)\, .\]
\end{remark}

\subsection{\SAT{} complexity hypotheses}
Below, we~state rigorously two 
popular hypotheses on~time complexity of~\SAT{}.

\begin{itemize}
    \item Strong exponential time hypothesis 
    (SETH)~\cite{DBLP:journals/jcss/ImpagliazzoPZ01, DBLP:journals/jcss/ImpagliazzoP01}: for every $\varepsilon>0$, there exists~$k$ such that
        \[\text{$k$-\SAT{}} \not \in \TIME[2^{(1-\varepsilon)n}] \, .\]
    \item Nondeterministic SETH (NSETH)~\cite{DBLP:conf/innovations/CarmosinoGIMPS16}: for every $\varepsilon>0$, there exists~$k$
    such that
    \[\text{$k$-\TAUT{}} \not \in \NTIME[2^{(1-\varepsilon)n}]\,.\] 
\end{itemize}

\cite{DBLP:journals/iandc/JahanjouMV18} proved that if SETH is false, then $\E^{\NP}$ requires series-parallel Boolean circuits of~size~$\omega(n)$. \cite[Corollary B.3]{DBLP:conf/innovations/CarmosinoGIMPS16} extended this result and showed that refuting NSETH is sufficient for such a~circuit lower bound.
\begin{theorem}[{\cite[Corollary B.3]{DBLP:conf/innovations/CarmosinoGIMPS16}}]\label{thm:nseth}
If NSETH is false, then $\E^{\NP}$ requires series-parallel Boolean circuits of size $\omega(n)$.
\end{theorem}

\subsection{Fine-grained reductions and SETH hardness}\label{sec:fine-grained-reductions}
Below, we~define fine-grained reductions and SETH-hardness. 
Since SETH-hardness is~defined as a~sequence of~reductions
from $k$-\SAT{} for every value of~$k$,
we~introduce a~function~$\delta$
in~the standard definition (\cite[Definition~6]{DBLP:conf/iwpec/Williams15}) of~fine-grained reductions.

\begin{definition}[Fine-grained reductions]\label{deg:fg}
    Let $P,Q$ be~problems, $p, q \colon \mathbb{Z}_{\ge 0} \to \mathbb{Z}_{\ge 0}$ be non-decreasing functions and ${\delta\colon\R_{>0}\to\R_{>0}}$.
    We say that $(P, p(n))$ \emph{$\delta$-fine-grained reduces} to~$(Q, q(n))$
    and write $(P,p(n)) \le_{\delta} (Q,q(n))$,
    if~for every $\varepsilon>0$ and $\delta=\delta(\eps)>0$, there exist 
    an~algorithm~$\mathcal A$ for~$P$ with oracle access to~$Q$,
    a~constant~$d$, a~function $t(n) \colon \mathbb{Z}_{\ge 0} \to \mathbb{Z}_{\ge 0}$, such that on~any instance of~$P$ of size~$n$, the algorithm~$\mathcal A$
    \begin{itemize}
        \item runs in time at~most $d(p(n))^{1-\delta}$;
        \item produces at~most $t(n)$~instances of~$Q$ adaptively: every instance depends on~the previously produced instances
        as~well as~their answers of~the oracle for~$Q$;
        \item the sizes $n_i$ of the produced instances (of~the problem~$Q$) satisfy the inequality
        \[\sum_{i=1}^{t(n)}q(n_i)^{1-\varepsilon} \le d(p(n))^{1-\delta} \, .\]
    \end{itemize}
\end{definition}
We say that $(P, p(n))$ \emph{fine-grained reduces} to~$(Q, q(n))$
    and write $(P,p(n)) \le (Q,q(n))$ if $(P,p(n)) \le_{\delta} (Q,q(n))$ for some function $\delta\colon\R_{>0}\to\R_{>0}$.

It is~not difficult to~see that
if $(P,p(n)) \le (Q,q(n))$, then any improvement over the running time $q(n)$ for the problem~$Q$
implies an~improvement over the running time~$p(n)$ for the problem~$P$: for any $\varepsilon>0$, there is $\delta>0$, such that if $Q$~can be~solved in time $O(q(n)^{1-\varepsilon})$, then $P$~can be~solved in time $O(p(n)^{1-\delta})$.

\begin{definition}[SETH-hardness]\label{def:seth-hardness}
For a~time bound $t \colon \mathbb{Z}_{>0} \to \mathbb{Z}_{>0}$, we say that a problem~$P$ is \emph{$t(n)$-SETH-hard} if there exists 
a~function $\delta\colon\R_{>0}\to\R_{>0}$ and for every $k\in\N$, 
\[(\text{$k$-\SAT{}},2^n)\leq_{\delta} (P,t(n)) \,.\]
\end{definition}
If a~problem~$P$ is $t(n)$-SETH-hard, then any algorithm 
solving~$P$ in time $t(n)^{(1-\varepsilon)}$ implies an algorithm solving $k$-\SAT{} in time $2^{(1-\delta(\eps))n}$ for all~$k$, thus, breaking SETH.

\subsection{Polynomial formulations}
Informally, a~polynomial formulation of a~computational problem~$A$ is an~explicit family of polynomials ${\mathcal P}=\{P_s\}_{s \in \mathbb N}$ and a~family of functions $\phi=(\phi_n)_{n\geq1}$ where $\phi_n\colon A_n \to \Z^{s(n)}$ 
such that to~check if $x\in A_n$ is a~yes-instance of~$A$, it suffices to map $y=\phi(x) \in \Z^{s(n)}$ and evaluate the corresponding polynomial $P_{s(n)}(y)$.
\begin{definition}[Polynomial formulations]
    Let $T \colon \mathbb N \to \mathbb N$ be a~time bound, and $\mathcal P=\{P_s\}_{s \in \mathbb N}$ be a~$\Delta$-explicit family of~polynomials over $\Z$. We say that $\mathcal P$
    is~a~\emph{$\Delta$-polynomial formulation of~$A$} of complexity~$T$, if there exist
    \begin{itemize}
        \item a~non-decreasing function $s \colon \mathbb N \to \mathbb N$ satisfying $s(n) \le T(n)$, and an algorithm computing $s(n)$ in~time $T(n)$;
        \item a~family of mappings $\{\phi_n\}_{n \in \mathbb N}$, where $\phi_n \colon A_n \to \Z^{s(n)}$, and an algorithm evaluating $\phi_n$ at any point in time $T(n)$
    \end{itemize}
    such that the following holds.
    For every~$n\in\N$ and every $x \in A_n$,
    \begin{itemize}
    		\item $P_{s(n)}(\phi_{n}(x)) \neq 0 \Leftrightarrow x \text{ is a~yes-instance of } $A$\,;$
                \item {$|\phi_n(x)| < 2^{(s(n))^{o(1)}}$. }
    \end{itemize}
\end{definition}

In~\cite{DBLP:conf/soda/BelovaGKMS23}, it~is proved that 
for any $\alpha>1$, there exists $\Delta=\Delta(\alpha)$ such that
$k$-\SAT{}, Hamiltonian Path, Graph Coloring, Set Cover,
Independent Set, Clique, Vertex Cover, and $3d$-Matching problems
admit a~$\Delta$-polynomial formulation of~complexity~$\alpha^n$.
In~this paper, 
for a~wide range of~problem from~\P{},
we~prove that for any $\alpha>0$
there exists $\Delta=\Delta(\alpha)$
such that the corresponding problem
admit a~$\Delta$-polynomial formulation 
of~complexity $n^{1+\alpha}$.

\section{New hypothesis (POSETH): solving \SAT{} with polynomial evaluation oracle}
\label{sec:poseth}

Carmosino et~al.~\cite{DBLP:conf/innovations/CarmosinoGIMPS16} have shown that proving an upper bound on~both nondeterministic and co-nondeterministic complexity of a~problem can be used to show the implausibility of SETH-based lower bounds. They also raised the question if studying some other computational models can be used to distinguish the fine-grained complexity 
of~various problems. We follow this suggestion and 
 introduce a~new computational resource $\POTIME$. 

\begin{definition}[\POTIME{} and \PONTIME{}]
Let $\mathcal{Q}=\{Q_n\}_{n=1}^{\infty}$ be~an~explicit family of~polynomials 
and $t \colon \mathbb{N} \to \mathbb{N}$ be~a~time bound. By~$\POTIME[t, \mathcal{Q}]$ ($\PONTIME[t, \mathcal Q]$)
we~denote the class of~problems that can be~solved
in~time~$t(n)$ by~a~deterministic (nondeterministic, respectively) algorithm that
makes oracle calls to~$\mathcal Q$:
for any $a_1, \dotsc, a_s \in \mathbb{Z}$,
such that  {$|a_i| = 2^{s^{o(1)}}$}    
for all~$i$,  
it~gets the value $Q_s(a_1, \dotsc, a_s)$ in~time~$s$.
\end{definition} 

The usefulness of this notion comes from the fact that fine-grained reductions transfer savings for \POTIME{}.

\begin{lemma}
    If $(A,p(n)) \le_{\delta} (B,q(n))$ and $B \in \POTIME(q(n)^{1-\epsilon},Q)$, then $A \in \POTIME(p(n)^{1-\delta},Q)$. 
\end{lemma} 

\begin{proof}
    We run the reduction from~$A$ to~$B$. Then we will substitute oracle calls to~$B$ from the reduction with the algorithm for~$B$ that can make oracle calls to~$Q$. 
    By the properties of the fine-grained reductions, we~have the desired running time. 
\end{proof}
 
Now we~introduce an~even stronger version of~SETH and show that refuting~it is~still challenging as~it~implies circuit lower bounds.

\begin{definition}[POSETH]
For every~$\epsilon>0$ and every explicit polynomial family~$\mathcal Q$, there exists~$k$ such that $k$-\TAUT{} cannot be~solved in time $2^{(1-\epsilon)n}$ by 
a~deterministic algorithm that makes oracle calls 
to~evaluating~$\mathcal Q$:
\[\forall \varepsilon>0,\, \forall 
\text{ explicit $\mathcal Q$},\, \exists k \colon \text{$k$-\TAUT{}}\not \in \POTIME[2^{(1-\epsilon)n}, \mathcal Q] .\]
\end{definition}

\begin{theorem}
    \label{thm:pnsethlowerbounds}
    If POSETH is~false, then either there exists a~family 
    of~Boolean functions
    from $\E^{\NP}$ that cannot be~computed 
    by~series-parallel Boolean circuits of~size $O(n)$ or~there exists
    an~explicit family of~polynomials
    that cannot be~computed by~arithmetic circuits of~size $n^{1+\alpha}$ for some $\alpha>0$.
\end{theorem}

\begin{proof}
    Assume that POSETH is~false:
    \[\exists \varepsilon>0,\, \exists
\text{ explicit~$\mathcal Q$},\, \forall k \colon \text{$k$-\TAUT} \in \POTIME[2^{(1-\epsilon)n}, \mathcal Q] .\]
    Let $\mathcal A$ be the corresponding deterministic algorithm with oracle access to~$\Delta$-explicit polynomial 
    family~$\mathcal Q$.
    Let $\alpha=\alpha(\varepsilon)$, $\beta=\beta(\varepsilon)$, $\gamma=\gamma(\varepsilon)$, and $\delta=\delta(\varepsilon)$
    be~positive parameters whose values will be~specified later. 
    Given a~$k$-CNF formula with $n$~variables, 
    $\mathcal A$~makes a~number of~oracle calls
    to~evaluating~$\mathcal Q$. 
    Let $t_1, \dotsc, t_m$ denote
    the size of~$m$~oracle calls (hence, the $i$-th call evaluates $Q_{t_i}$). Since the running time of~$\mathcal A$ is at~most $O(2^{(1-\epsilon)n})$,
    we~have that
    \[t_1+\dotsb+t_m \le 2^{(1-\epsilon)n} .\]

    If~$\mathcal Q$ cannot be~computed by~arithmetic circuits of~size $O(n^{1+\alpha})$, then we~are done. Assume that $\asize(Q_n) \le O(n^{1+\alpha})$. 
    Below, we~describe 
    a~nondeterministic algorithm
    that guesses the corresponding circuit, verifies~it, and uses~it to~solve $k$-\TAUT{} in~time $O(2^{(1-\delta)n})$ for all~$k$. By~Theorem~\ref{thm:nseth}, this gives 
    Boolean circuit lower bounds. The corresponding algorithm consists of~two stages. Below, 
    we~describe each stage and estimate its running time. Then, 
    we~choose the values for the parameters~$\alpha$, 
    $\beta$,~and~$\gamma$ to~ensure that the total running time 
    is~at~most $O(2^{(1-\delta)n})$.

    \begin{description}
    \item[Preprocessing stage.]
    Find an~arithmetic circuit $C_t$ of~size $O(t^{1+\alpha})$ computing $Q_t$ for all $t \le 2^{\beta n}$ and the corresponding prime number~$p_t$ using Lemma~\ref{lemma:macptwo} (and Remark~\ref{rem:rho}).
    Since we~do this for all $t \le 2^{\beta n}$, the total time needed for this stage is at~most 
    \[O\left(2^{\beta n} \cdot 2^{\beta n(1+\alpha)} \cdot 2^{\beta n(2\Delta+1)}\right)=O\left(2^{\beta n (3+\alpha+2\Delta)}\right) \, .\]
    This is~at~most $O(2^{(1-\delta)n})$, if
    \begin{equation}\label{eq:guessing}
        \beta(3+\alpha+2\Delta)\le 1-\delta.    
    \end{equation}
    
    \item[Solving stage.]
    Take an~input formula~$F$ of~$k$-\TAUT{} on~$n$~variables
    and branch on~all but $\gamma n$ variables.
    For each of~the resulting $2^{(1-\gamma)n}$
    formulas~$F'$ on $\gamma n$ variables do~the following. Assume that $\mathcal{A}(F')$~makes 
    $m$~oracle calls to~evaluating 
    $Q_{t_1}, \dotsc, Q_{t_m}$. The running time 
    of~$\mathcal{A}(F')$ is at most $2^{(1-\epsilon)\gamma n}$. In~particular, 
    \begin{equation}\label{eq:oraclecallsizes}
        t_1+\dotsb+t_m \le 2^{(1-\epsilon)\gamma n}.    
    \end{equation}
    To~solve~$F'$, we~apply~$\mathcal A$,
    but each time $\mathcal A$ makes a~query 
    to~$Q_t$, we~use $C_t$ to~find the result.
    This is~possible (we~have a~precomputed circuit
    for any input size), if
    \begin{equation*}
        \max\{t_1, \dotsc, t_m\} \le 2^{\beta n}.
    \end{equation*}
    This is true, if
    \begin{equation}\label{eq:beta}
        (1-\epsilon)\gamma \le \beta.
    \end{equation}

    {To~compute the value of~$C_t(a_1, \dotsc, a_t)$, we~simply compute the value of~each gate. The value of~each gate can be~computed (from the values of~its two predecessors) in~time~
    {$O(\log^2 p_t)=O(\log^2 2^{t^{o(1)}})=t^{o(1)}$.} 
    Multiplying this by~the total number of~gates, we~get that the time it~takes
    to~compute $C_t(a_1, \dotsc, a_t)$ is
    {\[O(t^{1+\alpha}t^{o(1)})=O(t^{1+2\alpha}) \, .\]}}

    Now, we~estimate the running time of this stage. The running time of~solving~$F'$ is at~most
    \begin{align*}
    &\phantom{\le\,\,\,} O\left(2^{(1-\epsilon)\gamma n}+(t_1^{1+2\alpha}+\dotsb+t_m^{1+2\alpha}) \right)\tag{Jensen's inequality}\\
    &\le  O\left(2^{(1-\epsilon)\gamma n}+(t_1+\dotsb+t_m)^{1+2\alpha} \right)\tag{by \eqref{eq:oraclecallsizes}}\\
    &\le O\left(2^{(1-\epsilon)\gamma n}+
    2^{(1-\epsilon)\gamma(1+2\alpha)n} \right)\\
    &=O\left(2^{(1-\epsilon)\gamma(1+2\alpha)n}\right).
    \end{align*}
    Then, the running time of~solving all $2^{(1-\gamma)n}$ formulas, each on~$\gamma n$
    variables, is~at~most \[O(2^{((1-\epsilon)\gamma(1+2\alpha)+(1-\gamma))n}).\]
    To~make this at~most $O(2^{(1-\delta)n})$, it~suffices to~ensure that
    \begin{equation}\label{eq:f}
        \delta \le \gamma(1-(1-\epsilon)(1+2\alpha)).
    \end{equation}
    \end{description}

    Finally, we~set the parameters $\alpha, \beta, \gamma, \delta$ to~ensure that the running
    time of~the resulting algorithm is $O(2^{(1-\delta)n})$.
    To~guarantee this, it~suffices to~satisfy 
    inequalities~\eqref{eq:guessing}--\eqref{eq:beta}.
    First, we~choose a~small enough $\alpha>0$ such that $(1-\varepsilon)(1+2\alpha)<1$. 
    This ensures that~\eqref{eq:f} holds for small enough $\delta>0$. Then, we~choose a~small enough $\beta$ such that $\beta(3+\alpha+2\Delta)<1$.
    Then, \eqref{eq:guessing} holds for small enough~$\delta>0$. Finally, we~set $\gamma$
    to~be small enough to~ensure~\eqref{eq:beta}.
\end{proof}

\begin{remark}
    \label{rem:pn}
    It~is not difficult to see that the proof goes~by
    if~one replaces POSETH by~PONSETH in~the theorem above.
\end{remark}

\begin{theorem}
\label{thm:pf-potime}
    If~a~problem~$A$ admits a~polynomial
    formulation of~complexity $t(n)$, then $A \in \POTIME[O(t(n)), \mathcal{Q}]$ for some explicit polynomial family~$\mathcal{Q}$. 
\end{theorem}
\begin{proof}
    The corresponding algorithm for~$A$ uses its polynomial formulation
    $\mathcal P=\{P_n\}_{n \in \mathbb N}$ of~complexity $t(n)$ as an~oracle.
    Given an~instance $I \in A_n$, the algorithm computes the number~$s(n)$ 
    of~variables and
    their values in~time $O(t(n))$ and then asks the oracle to compute the value of~$P_{s(n)}$. Then, $I$
    $I$~is a yes-instance if~and only~if the oracle returns a~non-zero value.
\end{proof}

\begin{theorem}
\label{thm:potime-poseth}
    Let $A \in \POTIME[t(n)^{1 - \delta}, \mathcal{Q}]$ ($\PONTIME[t(n)^{1 - \delta}, \mathcal{Q}]$) for some $\delta > 0$ and some explicit polynomial family $\mathcal{Q}$. Then $A$ is not $t(n)$-SETH-hard under POSETH (PONSETH).
\end{theorem}
\begin{proof}
    Assume that $A$~is $t(n)$-SETH-hard and consider the reduction $(\text{$k$-\SAT{}}, 2^n)\leq_{\delta} (P,t(n))$. Since we can solve~$A$ with oracle for~$\mathcal{Q}$ in~deterministic (nondeterministic, respectively) time $t(n)^{1 - \delta}$, we can use the reduction and solve $k$-\SAT{}  with the same oracle in~deterministic (nondeterministic, respectively) time $2^{(1 - \varepsilon)n}$ for some $\varepsilon > 0$ and every~$k$ which breaks POSETH (PONSETH, respectively). 
\end{proof}

\section{Problems from~\P{} that can be~solved efficiently
with polynomial evaluation oracle}
\label{sec:polyp}

In~this section, 
we~show that a~wide range of~problems from~\P{}
admit a~$\Delta(\varepsilon)$-polynomial formulation of~complexity $O(n^{1+\varepsilon})$ (or $O((n + m)^{1 + \varepsilon})$ for graph problems),
for any $\varepsilon>0$. 
To~capture many such polynomial-time solvable problems,
we~define a~subclass of~problems from~\P{}
where the goal is to~find a~constant size
subset of an~input that can be~verified locally.

Informally, we~say that a~decision problem~$A$
belongs to~the class~\LS{},
if, for any instance $I \in A_n$, it is a~yes-instance if and only if
there is a~constant size set $Y \subseteq I$
and a~constant size set~$N \subseteq \overline{I}$
such that the pair $(Y,N)$
satisfies some
easily verifiable property.
This is~best illustrated with examples. The $3$-\SUM{} problem clearly satisfies this property as~the goal is~just to~find three 
elements of~the input whose sum is~equal to~zero. Another example
of~a~problem with this property is Induced $4$-Cycle: in~this case,
one needs to~check whether the input graph contains
four edges and two non-edges that span four nodes. On~the other hand, for the Hamiltonian Path problem it~is not immediate whether 
it~satisfies this property as~one looks for a~subset of~edges 
of~non-constant size. Another example of a~problem whose straightforward statement does not put~it into the considered class
is~$k$-Vertex Cover (where the goal is~to~check whether the input graph contains a~set of~$k$~nodes that cover all edges). In~this problem, one 
is~looking for a~set of nodes (or~edges) of~constant size,
but one cannot ensure that it~covers the whole graph just 
by~checking a~constant number of~non-edges.

One technical remark is~in~order before we~proceed to~the
formal definition of~the class of~problems. Throughout the whole paper, by~$n$ we~denote the natural complexity parameter
of a~problem that is~used to~measure both the size 
of an~instance and the running time of~an~algorithm
for this problem. 
In~this chapter however it~is important
that we~measure the running time of~graph algorithms
in~terms of $n+m$, the number of~nodes plus the number of~edges.
To avoid naming conflicts, below we~use an~additional parameter~$s$ to~denote the size of~an~instance.

We start with some preliminary definitions.
For every $n \in \mathbb{N}$, we introduce a~\emph{universe} $U_n$, which is a set of size $n^{O(1)}$ that can be generated in polynomial time with a bijection $B_n \colon U_n \rightarrow [|U_n|]$ that can be computed in time~$O(1)$ for any element from $U_n$. 
For some fixed positive integers $\alpha$ and $\beta$, we introduce a~\emph{verification algorithm} $V: \mathbb{N}^{\alpha + \beta} \rightarrow \{0, 1\}$ that works in polynomial time. For $u_1, \dotsc, u_{\alpha+\beta} \in U_n$, by $V(u_1, \dotsc, u_{\alpha+\beta})$ we denote $V(B_n(u_1), \dotsc, B_n(u_{\alpha+\beta}))$. 

The common form of a~problem from \LS{} is a triple $(\alpha, \beta, V)$ and a~family $\{U_n\}_{n \in \mathbb N}$. It is 
a~decision problem accepting a~tuple $(n, m, S)$, where $S$ is a~subset of $U_n$ of size $m$, as an input and checking whether there exist a subset $X$ of $S$ of size $\alpha$ and a subset $Y$ of $U_n \setminus S$ of size $\beta$ such that $V$ returns $1$ on $(X, Y)$.

Before we finally define the class, let us illustrate that definition on the Induced 4-Cycle problem. 
For this problem, $\alpha = 4$, $\beta = 2$, corresponding to the number of edges and non-edges of 4-Cycle, respectively.
Translatable universe $U_n$ will be $[n]\times[n]$ that stands for all pairs of nodes and $B_n$ will be a~natural bijection that maps $[n] \times [n]$ into $[n^2]$ in such a way that for all $u \in [n - 1] \times [n - 1]$, $B_n(u) = B_{n - 1}(u)$. 
The verification algorithm~$V$ will check whether edges $x_1, \dots, x_4 \in U_n$ and non-edges $y_1, y_2 \in U_n$ span a 4-cycle. 

Now we can formulate Induced 4-Cycle as follows. Given 
a~graph $G$ as $(n, m, S)$, where $n$ is the number of nodes, $m$ is the number of edges and $S \subseteq U_n$ is a set of edges, check whether there exist $x_1, \dots, x_4 \in S$ and $y_1, y_2 \in U_n \setminus S$ such that $V(x_1, \dots, x_4, y_1, y_2) = 1$.

Note that we run $V$ on the result of encoding of pairs of nodes, and $V$ does not depend on $n$, so we have to ensure that $V$ interprets same pairs of nodes the same way even for different $n$. Fortunately, when we deal with cartesian products of integer intervals (such as $[n]\times[n]$), such a bijection always exists.

Now we are ready to define the class \LS{} formally.

\begin{definition}[Local Subset class] 
    The class \LS{} contains all
    decision problems~$A$ for which
    there exist positive integers $\alpha$~and~$\beta$,
    a family of universes $\{U_n\}_{n \in \mathbb N}$ ($|U_n|=n^{O(1)}$)
    and a~verification algorithm $V$ such that 
    for every instance $(n, m, S) \in A_{n + m}$ where $n, m \in \mathbb N$ and $S$ is a subset of $U_n$ of size $m$, $(n, m, S)$ is a~yes-instance if and only if there exist $x_1, \dotsc, x_{\alpha} \in S$ and $y_1, \dotsc, y_{\beta} \in U_n \setminus S$ such that
    $V(x_1, \dotsc, x_{\alpha}, y_1, \dotsc, y_{\beta})=1$.

    The size of an instance is defined as $n + m$, so $A_s = \{(n, m, S) \colon n + m = s, S \subseteq U_n, |S| = m\}$.
\end{definition}

    

It~is not difficult to~see that $\LS{} \subseteq P$. Indeed, if 
$A \in \LS{}$, then it~can be~solved by~the following polynomial-time algorithm. On an~input $I = (n, m, S) \in A_{n + m}$, generate a~set $U_n$; then, enumerate all $x_1, \dotsc, x_{\alpha} \in S$ and all $y_1, \dotsc, y_{\beta} \in U_n \setminus S$ and check whether $V(x_1, \dotsc, x_{\alpha}, y_1, \dotsc, y_{\beta})=1$.

\begin{theorem}
    Each of~the following problems belong to~the class \LS{}: 
    $k$-\SUM{}, 
    Collinearity,
    $H$-induced subgraph for any fixed graph~$H$,
    $\mathcal{H}$-induced subgraph for any fixed finite set $\mathcal{H}$ of graphs,
    the decision version of minimum weight $k$-clique,
    the decision version of MAX $H$-SUBGRAPH (with weights on vertices or on edges).
\end{theorem}
\begin{proof}
    An instance of a problem from the class \LS{} is of the form $(n, m, S)$. For the sake of simplicity, in this proof, we omit $n$ and $m$ in the problems descriptions.
    Recall that in all of~the following problems, 
    $R=\{-n^{O(1)}, \dotsc, n^{O(1)}\}$
    denotes the range 
    of~possible values of~the input integers.
    In~graph problems, we~assume that the set 
    of~nodes of~an~input graph is~$[n]$.
    For several problems, we fix the values of some coordinates of an input for the algorithm $V$, by that we mean that $V$ also checks those coordinates.
    \null\hfill\\
    \begin{itemize}
        \item $k$-\SUM{}. Let $U_n=[k] \times R$, $\alpha=k$, $\beta=0$. An~instance of~the problem is a~list of $kn$ pairs: 
        \[(1, a_1^1), (1, a_2^1), \dotsc, (1, a_n^1), (2, a_1^2),\dotsc, (2, a_n^2), \dotsc, (k, a_1^k), \dotsc, (k, a_n^k) \, .\]
        Given a~list $((1, b_1), \dotsc, (k, b_k))$ of~$\alpha=k$ elements from the input, the algorithm~$V$ checks that $b_1+\dotsc+b_k=0$.
        
        \item Collinearity. Let $U_n=R^2$, $\alpha=3$, $\beta=0$.
        Given $\alpha=3$ points $p_1, p_2, p_3$, the algorithm~$V$
        checks whether they are different and collinear.
        
        \item $H$-induced subgraph. Let $U_n=[n]^2$ and let $\alpha$~and~$\beta$ be the number of~edges and non-edges of~$H$, respectively. An~instance of~the problem 
        is~the set of~edges of the input graph~$G$.
        Given a~list $((u_1, v_1), \dotsc, (u_{\alpha}, v_{\alpha}))$ of~$\alpha$~edges and a~list 
        $((s_1, t_1), \dotsc, (s_{\beta}, t_{\beta}))$
        of~$\beta$~non-edges of~$G$, the algorithm~$V$ checks that they
        span a~graph with $|V(H)|$ nodes that is~isomorphic
        to~$H$.
        
        \item $\mathcal{H}$-induced subgraph. Let $U_n=[n + 1]^2$ and let $\alpha$~and~$\beta$ be the maximum number of~edges of a graph from $\mathcal{H}$ and the maximum number of non-edges of a graph from $\mathcal{H}$, respectively. An~instance of~the problem 
        is~the set of~edges of the input graph~$G$ with $n$ vertices, and an edge $(n + 1, n + 1)$. We add an additional vertex with a loop to ensure that $S, (U_n \setminus S) \neq \emptyset$.
        Given a~list $((u_1, v_1), \dotsc, (u_{\alpha}, v_{\alpha}))$ of~$\alpha$~edges and a~list 
        $((s_1, t_1), \dotsc, (s_{\beta}, t_{\beta}))$
        of~$\beta$~non-edges of~$G$, the algorithm~$V$ checks that 
        there exists a graph $H \in \mathcal{H}$ such that 
        edges $((u_1, v_1), \dotsc, (u_{|E(H)|}, v_{|E(H)|}))$ together with non-edges $((s_1, t_1), \dotsc, (s_{|\overline{E}(H)|}, t_{|\overline{E}(H)|}))$
        span a~graph with $|V(H)|$ nodes that is~isomorphic
        to~$H$ and contains only vertices from $[n]$.
        
        \item Minimum weight $k$-clique (decision version). 
        Let $U_n = [2] \times n \times n \times R$ and let $\alpha = \binom{k}{2} + 1$ and $\beta = 0$.
        An instance of the problem is a set $S = \{(1, u, v, w_{u, v}) \colon (u, v) \in E(G)\} \cup \{(2, W)\}$, describing the weighted edges of a graph $G$ and a proper encoded upper bound $W$ on the total weight of the $k$-clique. 
        Given a list $((1, u_1, v_1, w_{1}), \dots, (1, u_{\binom{k}{2}}, v_{\binom{k}{2}}, w_{\binom{k}{2}}), (2, W))$, the algorithm $V$ checks that the edges span a $k$-clique and $\sum w_{i} \le W$. The first coordinate helps us to distinguish edges of $G$ and $W$.
        
        \item MAX $H$-SUBGRAPH (decision version, weights on edges). Let $U_n = [3] \times n \times n \times R$, let $\alpha$ be the number of edges of $H$ plus 1, and let $\beta$ be the number of non-edges of $H$. An instance of the problem is a set $S = \{(1, u, v, w_{u, v}), (2, u, v) \colon (u, v) \in E(G)\} \cup \{(3, W)\}$, describing the weighted edges of a graph $G$ and a proper encoded lower bound $W$ on the total weight of the induced $H$. Given a list $((1, u_1, v_1, w_1), \dots, (1, u_{\alpha - 1}, v_{\alpha - 1}, w_{\alpha - 1}))$ of weighted edges of $G$, weight $(3, W)$ and a list $((2, s_1, t_1), \dotsc, (2, s_{\beta}, t_{\beta}))$ of non-edges of $G$, the algorithm $V$ checks that those edges and non-edges span a graph isomorphic to $H$ and that $\sum w_{i} \ge W$. 

        \item MAX $H$-SUBGRAPH (decision version, weights on vertices) 
        Let $U_n = [3] \times n \times R$, let $\alpha$ be the total number of vertices and edges in $H$ plus 1, and let $\beta$ be the number of non-edges of $H$. 
        An instance of the problem is a set $S = \{(1, u, v) \colon (u, v) \in E(G)\} \cup \{(2, v, w_w) \colon v \in V(G)\} \cup \{(3, W)\}$, describing the edges of $G$, the weights of vertices of $G$ and the lower bound $W$ on the total weight of the induced $H$. 
        Given a list $((1, u_1, v_1), \dots, (1, u_{|E(H)||}, v_{|E(H)|}))$ of edges of $G$, a list $((2, a_1, w_1), \dots, (2, a_{|V(H)|}, w_{|V(H)|})$ of vertices of $G$ with their weights, weight $(3, W)$ and a list $((1, s_1, t_1), \dotsc, (1, s_{\beta}, t_{\beta}))$ of non-edges of $G$, the algorithm $V$ checks that those edges and non-edges span a graph isomorphic to $H$ on vertices $a_1, \dots, a_{|V(H)|}$, and that $\sum w_{i} \ge W$. 
    \end{itemize}
\end{proof}

Before proving that every problem from \LS{} admit 
a~polynomial formulation of~size 
$(n + m)^{1 + \varepsilon}$, we~provide an~instructive
example.
Consider the triangle detection problem.
For an input graph~$G$, we assume that its vertex set $V = [n]$ and its edge set $E = \{e_1, \dotsc, e_m\} \subseteq [n]\times[n]$. To solve this problem, we can iterate over all tuples $(v_1, v_2, v_3) \in [n]^3$ of different vertices and all tuples $(i_{12}, i_{23}, i_{31}) \in [m]^3$, and check whether $e_{i_{12}} = (v_1, v_2)$, $e_{i_{23}} = (v_2, v_3)$ and $e_{i_{31}} = (v_3, v_1)$. It can be expressed as the following polynomial: 
\[\sum\limits_{\substack{v_1, v_2, v_3 \in [n] \\ i_{12}, i_{23}, i_{31} \in [m]}} 
x^{i_{12}}_{v_1, v_2} x^{i_{23}}_{v_2, v_3} x^{i_{31}}_{v_3, v_1} \; ,\]  
where a~variable $x^{i}_{u_1, u_2}$ is $1$ if $e_{i} = (u_1, u_2)$, and~$0$, otherwise. The issue with this polynomial is that it requires $mn^2$ variables, and we are limited with only $O((n + m)^{1 + \varepsilon})$. To solve this, we consider binary representation of the pairs of vertices, break them into $\theta = O(1)$ blocks and instead of checking that $e_i = (u_1, u_2)$, 
for every $q \in [\theta]$, we independently check that the $q$-th block of $e_i$ is equal to the the $q$-th block of $(u_1, u_2)$.
Let us denote the $q$-th block of the binary representation of $e$ as $e(q)$. Then, we can rewrite our polynomial in the following way: 
\[\sum\limits_{\substack{v_1, v_2, v_3 \in [n] \\ i_{12}, i_{23}, i_{31} \in [m]}} 
\prod\limits_{q = 1}^{\theta} 
x^{i_{12}, q}_{(v_1, v_2)(q)} \cdot x^{i_{23}, q}_{(v_2, v_3)(q)} \cdot x^{i_{31}, q}_{(v_3, v_1)(q)} \; ,\] 
where $x^{i, q}_{a}$ is 1 if $e_i(q) = a$, and~$0$, otherwise. The degree of the polynomial is still constant, but the number of variables is now $O(m \cdot \theta \cdot 2^{2 \log n / \theta}) = O(m n^{2 / \theta}) = 
O((n + m)^{1 + \varepsilon})$ for $\theta > 2 / \varepsilon$, and we can compute all of them also in time $O((n + m)^{1 + \varepsilon})$.
The only remaining issue is that we could construct different polynomials for $A_s$ since $s = n + m$ for different variations of $n$ and $m$. To resolve that, we can change $[n]$ and $[m]$ to $[n + m]$ in the polynomial and make the variables also verify that $i_{12}$, $i_{23}$, and $i_{31}$ are in fact in $[m]$.

We can use that technique for every \LS{} problem with $\beta = 0$. For example, for 
detecting a~(not necessarily induced) $4$-cycle,
we can construct the following polynomial (we use addition modulo 4):
\[\sum\limits_{\substack{v_1, \dots, v_4 \in [s] \\ i_{12}, i_{23}, i_{34}, i_{41} \in [s]}}
\prod\limits_{\ell = 1}^4 \prod\limits_{q = 1}^{\theta}  x^{i_{\ell, \ell + 1}}_{v_{\ell}, v_{\ell + 1}(q)} \; .
\]
However, if we want to check the existence of an \emph{induced} 4-cycle, we also need to verify non-edges. Since the number of non-edges can be large, we cannot iterate over all of them. To resolve that issue, we consider the edge set of $G$ as a sorted list of edges $e_1 < \dotsb < e_m$. Then, $(u, v)$ is a non-edge if $(u, v) < e_1$, $(u, v) > e_m$ or there exists $j \in [m - 1]$ such that $e_j < (u, v) < e_{j + 1}$. That means that to verify a non-edge it is sufficient to also iterate over the edges of $G$ and check some inequalities, that we also can do by dividing their binary representations into $\theta$ blocks. As a result, we obtain a polynomial formulation of complexity $O((n + m)^{1 + \varepsilon})$. 

\begin{theorem}
    For every $A \in \LS{}$ and every $\varepsilon>0$, there exist a~$\Delta$-polynomial formulation of complexity  $O(s^{1+\epsilon})$ of~$A$.
\end{theorem}

\begin{proof}
The idea is the following. We go through all possible
$a_1, \dotsc, a_{\alpha}$, $b_1, \dotsc, b_{\beta} \in U_n \colon$ $V(a_1, \dotsc, a_{\alpha}, b_1, \dotsc, b_{\beta}) = 1$ and verify that every $a_i \in S$ and every $b_i \notin S$. So we obtain a polynomial of the following form:
\[
\sum\limits_{\substack{
a_1, \dots, a_{\alpha} \in U_n, \\
b_1, \dots, b_{\beta} \in U_n \colon \\
V(a_1, \dots, a_{\alpha}, b_1, \dots b_{\beta}) = 1
}}
\prod\limits_{\ell = 1}^{\alpha} P^{\in}_{a_{\ell}}(X)
\prod\limits_{\ell = 1}^{\beta} P^{\notin}_{b_{\ell}}(X),
\]
where $P^{\in}_{a}(X)$ is the polynomial that verifies that $a \in S$, and $P^{\notin}_{b}(X)$ verifies that $b \notin S$.

Notice that we cannot introduce a separate variable for every element of $U_n$ because of its size. 
In order to resolve that problem, we do the following. We can assume that $U_n = [n^r]$ and $S = \{s_1, \dots, s_m\}$, where $m \le n$, $\forall i \in [m] \; s_i \in [n^r]$ and $s_1 < \dots < s_m$. We also consider $s_0 = 0$ and $s_{m + 1} = n^r + 1$. Then, for every $a \in S$ there exists $1 \le i \le m$ such that $a = s_i$, and for every $b \notin S$ there exists $0 \le j \le m$ such that $s_j < b < s_{j + 1}$.

Now, let us also iterate over all possible $1 \le i_1, \dots, i_{\alpha} \le m$ and $0 \le j_1, \dots, j_{\beta} \le m$ and check, whether each $a_{\ell} = s_{i_{\ell}}$ and each $s_{j_{\ell}} < b_{\ell} < s_{j_{\ell} + 1}$. The desired polynomial is now of the following form:
\[
\sum\limits_{\substack{
1 \le a_1, \dots, a_{\alpha} \le n^r, \\
1 \le b_1, \dots, b_{\beta} \le n^r \colon \\
V(a_1, \dots, a_{\alpha}, b_1, \dots b_{\beta}) = 1
}}
\sum\limits_{\substack{
1 \le i_1, \dots, i_{\alpha} \le m, \\
0 \le j_1, \dots, j_{\beta} \le m
}}
\prod\limits_{\ell = 1}^{\alpha} P^=_{i_{\ell}, a_{\ell}}(X)
\prod\limits_{\ell = 1}^{\beta} P^<_{j_{\ell}, b_{\ell}}(X)
\prod\limits_{\ell = 1}^{\beta} P^>_{j_{\ell} + 1, b_{\ell}}(X),
\]
where $P^{comp}_{i, a}(X)$ verifies that the result of comparison of $s_i$ with $a$ is $comp$.

Though this polynomial can be~used to~solve the problem, it 
is not yet a~polynomial formulation, as it depends on~$n$ and~$m$, and not on the instance size $s = n + m$. Therefore, we change $n$~and~$m$ to $n + m$ in the polynomial and make $P^{comp}_{i, a}(X)$ also verify that $i \le m + 1$ ($a_i \le n^r$ and $b_i \le n^r$ are automatically verified in $P^=_{i_{\ell}, a_{\ell}}(X)$ and  $P^>_{j_{\ell} + 1, b_{\ell}}(X)$, respectively).

It is left to construct polynomials $P^{comp}_{i, a}(X)$ of constant degree such that all their variables can be computed in time $O((n + m)^{1 + \varepsilon})$.

Let $\overline{s_i}$ and $\overline{a}$ be binary representations of $s_i$ and $a$, respectively.
The length of those two strings is $r \log n$. Let us split both of the strings into $\theta$ parts of length $r \log n / \theta$, where $\theta$ is a constant which we will define later. 

Let $\overline{s_i} = s_i^1, \dots, s_i^{\theta}$ and $\overline{a} = a^1, \dots, a^{\theta}$, and let $c^1, \dots, c^{\theta} \in \{<, =, >\}$ be the results of comparison of the corresponding $s_i^q$ and $a^q$. Notice that, having $c^1, \dots, c^{\theta}$, we can conclude the result of comparison of $s_i$ and $a$. 

Let $C_{\theta}^{comp}$ be a set of all tuples $(c^1, \dots, c^{\theta})$ that give $comp$ as an overall result of comparison. Formally, we define those sets as follows:
\[C^{=}_{\theta} := \{=\}^{\theta}, \quad
C^{<}_{\theta} := \bigcup\limits_{q = 1}^{\theta} (\{=\}^{q - 1} \times \{<\} \times \{<, =, >\}^{\theta - q}), \quad
C^{>}_{\theta} := \{<, =, >\}^{\theta} \setminus C^{=}_{\theta} \setminus C^{<}_{\theta} .\]

Now we can iterate through all $(c^1, \dots, c^{\theta}) \in C^{comp}_{\theta}$ and for every $1 \le q \le \theta$ independently check, whether the result of comparison of $s_i^q$ and $a^q$ is $c^q$.

The final polynomial $P^{comp}_{i, a}(X)$ is the following.
\[
P^{comp}_{i, a}(X) = \sum\limits_{(c^1, \dots, c^{\theta}) \in C^{comp}_{\theta}}
\prod\limits_{q = 1}^{\theta} x^{c_q}_{i, q, a^q},
\]
where \[X = \{x^c_{i, q, a} \colon 
    c \in \{<, =, >\},\, 
    0 \le i \le n + m + 1,\,
    1 \le q \le \theta,\,
    a \in \{0, 1\}^{r \log (n + m) / \theta}\},
\] 
and $x^c_{i, q, a} = 1$ if~and only~if $i \le m + 1$ and the result of~comparison of $s_i^q$ and~$a$ is~$c$.

The number of variables is $3 (n + m + 1) \theta (n + m)^{r / \theta} = {O}((n + m)^{1 + r / \theta})$.
Each variable can be computed in time $O(1)$.
Take $\theta > r / \varepsilon$.
Then, all variables can be computed in time ${O}((n + m)^{1 + \varepsilon})$.
The polynomial degree is $\Delta=\theta (\alpha + 2\beta)=O(1)$. 
The whole polynomial can be constructed in time $(n + m)^{{O}(1)}$.
\end{proof}



\begin{corollary}
\label{cor:pf}
    For every $\varepsilon > 0$, the following problems admit polynomial formulations of complexity $s^{1 + \varepsilon}$ (where $s$ is the size parameter of the problem):
    \begin{itemize}
        \item $k$-\SUM{},
        \item Collinearity,
        \item $H$-induced subgraph for any fixed graph~$H$,
        \item $\mathcal{H}$-induced subgraph for any fixed finite set $\mathcal{H}$ of graphs,
        \item the decision version of minimum weight $k$-clique,
        \item the decision version of MAX $H$-SUBGRAPH (with weights on vertices or on edges).    
    \end{itemize}
\end{corollary}

\begin{corollary}
    For every $\varepsilon > 0$, the problems listed above are not $s^{1 + \varepsilon}$-SETH-hard under POSETH.
\end{corollary}
\begin{proof}
    Let a problem $A$ from the list be $s^{1 + \varepsilon}$-SETH-hard for some $\varepsilon > 0$. 
    Let us choose $0 < \varepsilon' < \varepsilon$.
    By Corollary \ref{cor:pf}, $A$ admits a polynomial formulation of complexity $s^{1 + \varepsilon'}$. Then, by Theorem \ref{thm:pf-potime}, $A \in \POTIME[O(s^{1 + \varepsilon'}), \mathcal{Q}]$ for some explicit polynomial family $\mathcal{Q}$ and we get a contradiction with Theorem \ref{thm:potime-poseth}.
\end{proof}

\section{Computations with small polynomial evaluation oracle: connections between nondeterministic complexity and arithmetic circuit complexity}
\label{sec:polyar}

Recall that in~\cite{DBLP:conf/soda/BelovaGKMS23}, it~is proved that 
for any $\alpha>1$, there exists $\Delta=\Delta(\alpha)$ such that
$k$-\SAT{}, \MAX{}-$k$-\SAT{}, Hamiltonian Path, Graph Coloring, Set Cover,
Independent Set, Clique, Vertex Cover, and $3d$-Matching problems
admit a~$\Delta$-polynomial formulation of~complexity~$\alpha^n$.
That~is, at~the expense of~increasing the degree of a~polynomial,
one can make the number of~variables in~this polynomial 
an~arbitrary small exponential function. In~this section, we~generalize this result and consider computations with polynomial evaluation oracle whose all oracle calls compute a~polynomial
that is not too~large.
We~show that 
\MAX{}-$k$-\SAT{} (for $k \ge 3$), 
{Set Cover with $n^{O(1)}$ sets of size at most $O(\log{n})$}, and 
Binary Permanent problems
can be~solved in~deterministic time $2^{(1-\varepsilon)n}$
by an~algorithm that makes oracle calls to~evaluating a~polynomial
whose number of~variables is~not too large. 
Using this, we~show that if~any of~these problem is not in~$\coNTIME[O(2^{(1-\delta)n})]$
for any~$\delta>0$, then, for any~$\gamma$, there exists an~explicit polynomial family that cannot be~computed by~arithmetic circuits of~size $O(n^\gamma)$.
It~should be~noted that, currently, for none
of~these problems it~is known how to~solve 
them faster than in~co-nondeterministic time~$2^n$.

\begin{theorem}
    \label{thm:qa}
    Let $A$ be~any of~the following problems:
        \MAX-$k$-\SAT{} (for $k \ge 3$), 
        {Set Cover with $n^{O(1)}$ sets of size at most $O(\log{n})$}, 
        or
        Binary Permanent.
    If $A$~cannot be~solved in~co-nondeterministic time $O(2^{(1-\delta)n})$, for any $\delta>0$, then, for any~$\gamma$, there exists an~explicit polynomial family that cannot be~computed by~arithmetic circuits of~size $O(n^{\gamma})$.
\end{theorem}

It~is interesting to~compare this result with 
a~result by~\cite{DBLP:journals/iandc/JahanjouMV18}
that says that refuting SETH implies Boolean circuit lower bounds:
both results derive circuit lower bounds, but 
\cite{DBLP:journals/iandc/JahanjouMV18} derives it~from
time upper bounds, whereas we~derive~it
from time lower bounds.

We~derive Theorem~\ref{thm:qa} from the following theorem.

\begin{theorem}
    \label{thm:nondetarith}
    Assume that a~problem~$A$ satisfies the following property:
    \begin{quote}
        for any $0<\beta<1$, there exists an~explicit family of~polynomials~$\mathcal Q$, such that for any $0<\varepsilon<1/2$, there exists a~deterministic algorithm
        that solves~$A$ in~time $O(2^{(1-\varepsilon)n})$
        and makes oracle calls to~computing $Q_s(x_1, \dotsc, x_s)$
        where $s \le 2^{\beta \epsilon n}$ and {$|x_i| \le 2^{O(s^2)}$} 
        for all $i \in [s]$.
    \end{quote}
    Then, if $A \not \in \coNTIME[O(2^{(1-\delta)n})]$, for any $\delta>0$, then, for any~$\gamma$, there exists an~explicit polynomial family that cannot be~computed by~arithmetic circuits of~size $O(n^\gamma)$.
\end{theorem}
\begin{proof}
    {First, note that if the property below holds for a problem $A$, then it holds for $\neg{A}$.}
    Assume that there exists~$\gamma$ such that every explicit family 
    of~polynomials has arithmetic circuit size~$O(n^\gamma)$. Below, 
    we~show that in~this case {$A$~can be~solved in~co-nondeterministic time $O(2^{(1-\delta)n})$ (or $\neg{A}$~can be~solved in~nondeterministic time $O(2^{(1-\delta)n})$)} for some $\delta=\delta(\gamma)$. To~do this, 
    we~take the corresponding algorithm with polynomial evaluation
    oracle {for $\neg{A}$}, nondeterministically guess arithmetic circuits for the corresponding polynomial, and replace every oracle call 
    with a~call to~the corresponding circuit.

    Let $0<\beta<1$ be~fixed and $\mathcal Q$~be the corresponding $\Delta$-explicit polynomial family. 
    For a~parameter~$\alpha$ to~be chosen later, for all $t \le 2^{\alpha n}$, find a~circuit $C_t$ of size $O(t^{\gamma})$ computing $Q_t$
    (for all points $(x_1, \dotsc, x_t)$ such that $|x_i| \le 2^{O(t^2)}$ for all $i \in [t]$). By~Lemma~\ref{lemma:macptwo}, it~can be~done in~nondeterministic time
    {
    \[O\left( 2^{\alpha n} \cdot 2^{\alpha \gamma n} \cdot 2^{2\alpha\Delta n} \cdot 
    \log^2\left( 2^{\alpha n} \cdot 2^{O\left(2^{2\alpha n}\right)}\right)
    \right)=O\left( 2^{\alpha n(5+\gamma+2\Delta)}\right) \, .\]}
    Take $\alpha=\frac{1}{6(5+\gamma+2\delta)}$ to~ensure that this 
    is~at most~$O(2^{n/2})$.

    Now, take $\varepsilon=\alpha/\beta$ (clearly, $0 < \varepsilon < 1/2$) and consider the algorithm from the problem statement that solves~$A$ in time $O(2^{(1-\varepsilon)}n)$
    with oracle calls of~size at~most $2^{\beta \varepsilon n}=2^{\alpha n}$. For each such oracle call~$Q_t$, we~have a~circuit $C_t$ of~size $O(t^{\gamma})$
    (that computes $Q_t$ correctly on~all inputs of~bounded absolute value). 
    
    Recall that the algorithm makes oracle calls for $Q_t(x_1, \dotsc, x_t)$ where $|x_i| \le 2^{O(t^2)}$ for all $i \in [t]$. 
    By~\eqref{eq:explicitpolyupper}, 
    $|Q_t(x_1, \dotsc, x_t)|<M_t$ 
    where $M_t=O((t\rho)^{2\Delta})={2^{O(t^2)}}$.
    To~compute the value of~$Q_t(x_1, \dotsc, x_t)$
    using the circuit~$C_t$, we~need to~bound the absolute value of~the intermediate results. To~this end,
    we~take a~prime~$2M_t \le p_t \le 4M_t$ and compute 
    the result of~every gate of $C_t(x_1, \dotsc, x_t)$
    modulo~$p_t$. One can generate~$p_t$
    in~nondeterministic time $O(\log^7M_t)=O(t^{14})$~\cite{aks04,lp19}. Generating such primes for all $t \le 2^{\alpha n}$ takes nondeterministic time at~most
    \[O\left(2^{\alpha n} \cdot 2^{14\alpha n}\right)={O(2^{15\alpha n})} \, .\]
    {Since $\alpha \le 1/30$, this is at~most $O(2^{n/2})$.}
    
    Using the circuit $C_t$ and the prime~$p_t$, one can compute the value of~$Q_t(x_1, \dotsc, x_t)$ in~time {$O(t^{\gamma+4})$}: the value (modulo~$p_t$) of~each of~$O(t^{\gamma})$ gates can be~computed (from the value 
    of~two predecessor gates) in~time $O(\log^2 p_t)=O(t^4)$.
    Thus, by~replacing every oracle call with a~call to~the corresponding circuit, we~get an~algorithm with running time
    \[O\left( 2^{(1-\varepsilon)n} \cdot 2^{\alpha (\gamma + 4) n}\right)=O(2^{n(1-\varepsilon+\alpha (\gamma + 4))}) \, .\]
    To~make this at~most $O(2^{(1-\delta)n})$,
    we~need to~ensure that 
    \[\varepsilon-\alpha (\gamma+4) > \delta \, .\]
    Plugging $\varepsilon=\alpha/\beta$, turns this inequality into
    \[\alpha \left( \frac{1}{\beta} - (\gamma+4) \right) > \delta \,.\]
    To~satisfy this, it~suffices to~take $\beta < 1/(\gamma + 4)$: then, the left-hand side is~positive and one can find the value of $\delta>0$
    such that the inequality is~satisfied.
\end{proof}


Now, we~are ready to~derive Theorem~\ref{thm:qa}
from Theorem~\ref{thm:nondetarith}.

\begin{proof}[Proof~of~Theorem~\ref{thm:qa}]
    For each of~the three problems from the theorem statement,
    we~show that it~satisfies the property from Theorem~\ref{thm:nondetarith}. Each of~the four proofs follows
    the same pattern.

    We~start by~showing that for any instance~$I \in A_n$
    and any $0<\alpha<1$, one can reduce~$I$ to~$2^{(1-\alpha)n}$
    instances of~size $\alpha n$ of~not necessarily the same problem~$A'$. 
    Then, we~show that for any $0 < \gamma < 1$, 
    the problem~$A'$
    admits a~polynomial formulation
    of~complexity $O(2^{\gamma n})$. This 
    allows~us to~construct the desired algorithm as~follows.

    Let $0<\beta<1$ be~fixed. Set 
    $\gamma=\frac{\beta}{\beta+1}$ (then, $\beta=\frac{\gamma}{1-\gamma}$)
    and consider the corresponding polynomial formulation~$\mathcal{Q}$ of~$A'$.

    Now, let $0 < \varepsilon < 1/2$ be fixed. 
    Set $\alpha=\frac{\varepsilon}{1-\gamma}$ 
    (then clearly $0 < \alpha < 1$).  
    Reduce an~instance $I \in A_n$ to $2^{(1-\alpha)n}$
    instance of~size~$\alpha n$ of the problem~$A'$.
    Solve each
    of~the resulting $2^{(1-\alpha) n}$ instances
    using an~oracle call to~$\mathcal Q$. Then, the running time of~this algorithm~is
    \[O\left(2^{(1-\alpha)n} \cdot 2^{\alpha \gamma n}\right)=
    O\left(2^{(1-\alpha(1-\gamma))n}\right)=O\left(2^{(1-\varepsilon)n}\right) \, .\]
    Also, the size of~each oracle call is~at~most
    \[2^{\alpha \gamma n}=2^{\frac{\varepsilon \gamma}{1-\gamma}n}=2^{\varepsilon \beta n}\,,\]
    as~desired.
    
\begin{itemize}
    \item \textbf{\MAX{}-$k$-\SAT{}.} For this problem everything is~particularly easy, as~it~is naturally self reducible (branch on~all but $\alpha n$ variables) and it~admits a~polynomial formulation of~complexity 
    $O(2^{\gamma n})$, for any $0<\gamma<1$, as~proved in~\cite{DBLP:conf/soda/BelovaGKMS23}.


    \item {\bf Set Cover with $n^{O(1)}$ sets of size at most $O(\log{n})$).} 
    Let the problem $A'$ be \#Set Partition problem: given a collection $\mathcal{S}$ of subsets of $[n]$ (with possible repetitions) find the number of partitions $\{S_1,\dots,S_k\} \subseteq \mathcal{S}$ such that $\bigsqcup_{i=1}^k S_i = [n]$.
    We first show that for every $\alpha > 0$ Set Cover can be reduced to $2^{(1 - \alpha)n}$ instances of \#Set Partition of size $\alpha n$ in time $O^*(2^{(1 - \alpha)n})$. 

    To do that, we introduce a more general problem \#Hybrid Cover Partition with parameters $n,m,k$ ($\#HCV_{n,m,k}$): given a collection $\mathcal{S}$ of subsets of $[n]$ (with possible repetitions) find the number of collections $\{S_1,\dots,S_k\} \subseteq \mathcal{S}$ such that $\bigcup_{i=1}^k S_i = [n]$ and each element in $[m]$ is present in exactly one of the $S_i$.  
    Note that \#Set Partition is exactly $\#HCV_{n,n,k}$.

    \begin{lemma}
        Set Cover with $n^{O(1)}$ sets of size at most $O(\log{n})$ can be reduced to $\#HCV_{n,m,k}$ for any $m\leq n$ with $n^{O(1)}$ sets of size $O(\log{n})$ in polynomial time.
    \end{lemma}
    \begin{proof}
        Let $\mathcal{S} = \{S_i\}_{i \in [n^{O(1)}]}$ be an instance of Set Cover.     
        We replace every set $S_i$ with $2^{|S_i \cap [m]|} = n^{O(1)}$ new sets: 
        $\{S_i^{T} := S_{i} \setminus T \colon T \subseteq S_i \cap [m]\}$, and obtain a collection $\mathcal{S}'$ which we consider as an instance of $\#HCV_{n,m,k}$. We show that $\mathcal{S}$ is a yes-instance of Set Cover iff $\#HCV_{n,m,k}(\mathcal{S}') > 0$.
    
        Let $S_{i_1}, \dots, S_{i_k}$ be a solution for $\mathcal{S}$, then $S_{i_1}, S_{i_2}^{[m] \cap S_{i_1}}, \dots, S_{i_k}^{[m] \cap (S_{i_1} \cup \dots \cup S_{i_{k - 1}})}$ is one of the solutions for $\mathcal{S}'$. In the backwards direction, if $S_{i_1}^{T_1}, \dots, S_{i_k}^{T_k}$ is a solution for $\mathcal{S}'$, then $S_{i_1}, \dots, S_{i_k}$ is a solution for $\mathcal{S}$.
    \end{proof}

    \begin{lemma}
        $\#HCV_{n,m,k}$ with $n^{O(1)}$ sets of size $O(\log(n))$ can be reduced to $2^{n-m}$ instances of \#Set Partition with the universe of size $m$ and $n^{O(1)}$ sets of size $O(\log(n))$ in time $O^*(2^{n - m})$. 
    \end{lemma}
    \begin{proof}
        Let $\mathcal{S}$ be an instance of $\#HCV_{n,m,k}$ and let us assume that $n > m$.  
        Then $\#HCV_{n,m,k}(\mathcal{S}) = \#HCV_{n-1,m,k}(\mathcal{S}_1) - \#HCV_{n-1,m,k}(\mathcal{S}_2)$, where $\mathcal{S}_1 = \{S \setminus \{n\}| S \in \mathcal{S}\}$,  $\mathcal{S}_2 = \{S | S \in \mathcal{S}, n \not\in S\}$.
        Using this equality we can reduce calculation of $\#HCV_{n,m,k}$ to two calculations of $\#HCV_{n - 1,m,k}$. After $n - m$ branchings we will end up with $2^{n - m}$ instances of $\#HCV_{m,m,k}$. Note that each collection of sets in those instances has at most $|\mathcal{S}| = n^{O(1)}$ sets of size at most $O(\log(n))$, and that $\#HCV_{m,m,k}$ is \#Set Partition with the universe of size $m$.
    \end{proof}

    Combining the two lemmas and setting $m = \alpha n$, we get that Set Cover can be reduced to $2^{(1 - \alpha)n}$ instances of \#Set Partition of size $\alpha n$ (note that the number of sets $n^{O(1)} = (\alpha n)^{O(1)}$, and the sizes of sets $O(\log n) = O(\log (\alpha n))$) in time $O^*(2^{(1 - \alpha)n})$.

    We now present a polynomial formulation for \#Set Partition of complexity $2^{\gamma n}$, for any $0 < \gamma < 1$.
    
    \begin{lemma}
        For every $0 < \gamma < 1$, \#Set Partition with the universe of size $n$ and $n^{O(1)}$ sets of size $O(\log(n))$ admits a polynomial formulation of complexity $2^{\gamma n}$.
    \end{lemma}
    \begin{proof}
        Let us fix the constant $0 < \gamma < 1$, and let $\theta$ be a large enough constant.
        Let $S_{i_1}, \dots, S_{i_k}$ be a partition of $[n]$, and let $S_{i_j}$ be sorted by the value of its minimum element. For now we consider only non-empty $S_{i_j}$. We can assume that every $|S_{i_j}| \le \frac{n}{2 \theta}$.
    
        Let $\ell_0 = 0$.
        Let $\ell_1$ be the first index such that $|\bigsqcup\limits_{j = 1}^{\ell_1} S_{i_j}| > n / \theta$.
        Let $A_1 := \bigsqcup\limits_{j = 1}^{\ell_1 - 1} S_{i_j}$ and $B_1 = S_{i_{\ell_1}}$.
        Let $\ell_2$ be the first index such that $|\bigsqcup\limits_{j = \ell_1 + 1}^{\ell_2} S_{i_j}| > n / \theta$.
        Let $A_2 := \bigsqcup\limits_{j = \ell_1 + 1}^{\ell_2 - 1} S_{i_j}$ and $B_2 = S_{i_{\ell_2}}$.
        Continuing this process we obtain sets $A_1, B_1, \dots, A_q, B_q$ for some $q \le 2\theta$.
        Let $k_j = \ell_j - \ell_{j - 1}$ for every $j \in [q]$.

        Note, that for every solution of \#Set Partition its partition $A_1 \sqcup B_1 \sqcup \dots \sqcup A_q \sqcup B_q = [n]$ is unique, so to find the number of solutions of \#Set Cover we can find number of solutions for every possible partition and sum them.
        
        For every $q \le 2\theta$ and every partition $A_1 \sqcup B_1 \sqcup \dots \sqcup A_q \sqcup B_q = [n]$ such that $\forall j$ $\frac{n}{2 \theta} \le |A_j| \le \frac{n}{\theta}$, $|B_j| \le \frac{n}{2 \theta}$, $|A_j \sqcup B_j| > \frac{n}{\theta}$ and all those sets are sorted by the value of their minimum elements,
        we construct a polynomial that calculates the number of solutions corresponding to that partition.
    
        Let us construct such a polynomial.
        We go through all possible $k_1, \dots, k_q$ such that $\sum k_j \le k$, and we add monomial $x_{k - \sum k_j} \cdot \prod\limits_{j \in [q]} y_{B_j} \cdot z_{A_j, B_j, k_j - 1}$, where $x_{k'}$ 
        is the number of ways to choose $k'$ empty sets, $y_{B}$ is the number of sets in 
            $\mathcal{S}$ that are equal to $B$, and $z_{A, B, k'}$ is the number of partitions of $A$ into $k'$ sets from $\mathcal{S}$ such that the minimum element of each used set is less than the minimum element in $B$.   
    
        \begin{lemma}
            Each variable $z_{A, B, k'}$ can be calculated in time $O^*(2^{n / \theta})$.
        \end{lemma}
        \begin{proof}
            Let $b$ be the minimum element of $B$. For every $X \subseteq A$ and $0 \le k \le k'$ we compute a dynamic programming $dp[X][k]$ that equals to the number of partitions of $X$ into $k$ sets from $\mathcal{S}$ such that the minimum element of each used set is less than $b$. Then, $z_{A, B, k'} = dp[A][k']$. 
            
            Let $dp[\emptyset][0] = 1$, let $dp[X][k] = 0$ if the minimum element $x$ of $X$ is greater 
            than $b$, and let $dp[X][k] = \sum\limits_{\substack{S \in \mathcal{S} \colon x \in S, S \subseteq X}} dp[X \setminus S][k - 1]$. This dynamic programming can be computed in time $O^*(2^{|A|}) = O^*(2^{n / \theta})$.
        \end{proof}
    
        There are $n^{O(1)}$ variables $x_{k'}$, each of which can be calculated in polynomial time. There are $O(\binom{n}{n / 2\theta})$ variables $y_{B}$ each of which can be also calculated in polynomial time. And there are $\binom{n}{n / \theta} \binom{n}{n / 2\theta} n^{O(1)}$ variables $z_{A, B, k'}$ each of which can be calculated in time $O^*(2^{n / \theta})$. So we can compute all the variables in time $O^*(2^{\gamma n})$.
    
        The final polynomial has degree at most $1 + 4 \theta$ and can be constructed in time $O^*(\binom{n}{n / \theta}^{2 \theta}) = O^*(2^{\theta n})$.
    \end{proof}

    \item \textbf{Binary Permanent.}

It's known that computation of Binary Permanent is closely related to counting the number of bipartite matchings. Consider a bipartite graph $G(L, R, E)$ with $(u, v) \in E \iff a_{u, v} = 1$, then the permanent of $A$ equals to number of perfect matchings in $G$.

Following the proof strategy of the previous theorems we will show how to reduce Binary Permanent to a number of smaller subproblems which admit polynomial formulations.

For bipartite graph $G(L, R, E)$, let $\Phi(G)$ be the set of mappings from $L$ to $R$ such that for every $\phi \in \Phi$ and $u \in L$, $(u, \phi(u)) \in E$. In particular, the number of bijective mappings in $\Phi$ is a number of perfect matchings in $G$ which is a permament of matrix $A$.

For a mapping $\phi$, we call $v \in R$ \emph{covered} if there is at least one $u \in L$ such that $\phi(u) = v$, and \emph{not covered}, otherwise.

To split the problem into smaller ones we introduce function $F(S^{=1}, S^{=0}, S^{\geq 1})$. For pairwise disjoint subsets $S^{=1}, S^{=0}$ and $S^{\geq 1}$ of $R$, $F(S^{=1}, S^{=0}, S^{\geq 1})$ is a set of mappings $\phi \in \Phi$ such that:
\begin{itemize}
    \item each $v \in S^{=1}$ is covered by exactly one $u \in L$ (i.e. there is exactly one $u \in L$ such that $\phi(u) = v$);
    \item each $v \in S^{=0}$ is not covered;
    \item each $v \in S^{\geq 1}$ is covered;
    \item each $v \in R \setminus S^{=1} \setminus S^{=0} \setminus S^{\geq 1}$ is either covered or not.
\end{itemize}

To reduce evaluation of the permanent to smaller subproblems we will prove the following lemma.

\begin{lemma}
    Let $|L| = |R| = n$, then for every $\alpha > 0$ evaluation of the number of perfect matchings in $G$ can be reduced to $2^{(1 - \alpha)n}$ evaluations of $|F(S^{=1}, S^{=0}_i, \emptyset)|$ with $|S^{=1}| = \alpha n$. Note, that for all $|F(S^{=1}, S^{=0}_i, \emptyset)|$, $S^{=1}$ is the same.
\end{lemma}
\begin{proof}
    Let us see how to reduce evaluation of $|F(S^{=1}, S^{=0}, S^{\geq 1})|$ to smaller subproblems. Consider any $v \in S^{\geq 1}$. We prove the following identity.  $F(S^{=1}, S^{=0}, S^{\geq 1}) = F(S^{=1}, S^{=0}, S^{\geq 1} \setminus \{v\}) \setminus F(S^{=1}, S^{=0} \cup \{v\}, S^{\geq 1} \setminus \{v\})$ or $|F(S^{=1}, S^{=0}, S^{\geq 1})| = |F(S^{=1}, S^{=0}, S^{\geq 1} \setminus v)| - |F(S^{=1}, S^{=0} \cup \{v\}, S^{\geq 1} \setminus \{v\})|$.

    $F(S^{=1}, S^{=0}, S^{\geq 1})$ is a number of mappings from $\Phi$ such that all $u \in S^{=1}$ are covered exactly once, all $u \in S^{=0}$ are not covered, all $u \in S^{\geq 1}$ are covered at least once, and there are no restrictions for $u \in R \setminus S^{=1} \setminus S^{=0} \setminus S^{\geq 1}$. $F(S^{=1}, S^{=0}, S^{\geq 1} \setminus \{v\})$ is similar but without any restrictions for $v$, so the difference between $F(S^{=1}, S^{=0}, S^{\geq 1} \setminus \{v\})$ and $F(S^{=1}, S^{=0}, S^{\geq 1})$ are the mappings from $\Phi$ satisfying common restrictions for $u \neq v$ and having $v$ uncovered which is the set $F(S^{=1}, S^{=0} \cup \{v\}, S^{\geq 1} \setminus \{v\})$.

    Let $S^{=1}$ be the arbitrary subset of $R$ of size $2^{\alpha n}$ and consider $F(S^{=1}, \emptyset, R \setminus S^{=1})$. This is the set of mappings in $\Phi$ which covers vertices from $S^{=1}$ exactly once and vertices from $R \setminus S^{=1}$ at least once. This is exactly the number of bijective mappings in $\Phi$ since these mappings are surjective (since every $v \in R$ belongs to either $S^{=1}$ or $S^{\geq 1}$) and $|L| = |R|$. So, the $|F(S^{=1}, \emptyset, R \setminus S^{=1})|$ is the number of perfect matchings in $G$. Now we can apply identity $|F(S^{=1}, S^{=0}, S^{\geq 1})| = |F(S^{=1}, S^{=0}, S^{\geq 1} \setminus v)| - |F(S^{=1}, S^{=0} \cup \{v\}, S^{\geq 1} \setminus \{v\})|$ $(1 - \alpha)n$ times. Every time the number of summands doubles, size of $S^{\geq 1}$ reduces by $1$ and $S^{=1}$ is the same in all summands, so after $(1 - \alpha)n$ steps we will end up with $2^{(1 - \alpha)n}$ summands with empty $S^{\geq 1}$ and same $S^{=1}$.
\end{proof}

Following the lemma, we can use evaluation of $|F(S^{=1}, S^{=0}_i, \emptyset)|$ as $A'$ problem, reducing permanent evaluation to $2^{(1 - \alpha) n}$ instances of $F$ in time $O^*(2^{(1 - \alpha) n})$. Now we show, that $F$ admits efficient polynomial formulation.

\begin{lemma}
    For every $0 < \gamma < 1$, there is a polynomial formulation for $|F(S^{=1}, S^{=0, i}, \emptyset)|$ with $|S^{=1}| = \alpha n$ of complexity $O(2^{\gamma n})$.
\end{lemma}

\begin{proof}
    Choose $\theta$ to be defined later.
    Consider some mapping $\phi \in F(S^{=1}, S^{=0}_i, \emptyset)$. For every $v \in S^{=1}$ there is exactly one $u \in L$ such that $\phi(u) = v$. We will denote $u$ as preimage of $v$ in $\phi$.

    To find the number of mappings in $F(S^{=1}, S^{=0}_i, \emptyset)$ we will distribute them into non-intersecting classes, find number of mappings in each class and then just sum these values.

    Consider some $\phi \in F(S^{=1}, S^{=0}_i, \emptyset)$, for each $v \in S^{=1}$ find $\phi^{-1}(v)$. Now every element of $L$ is either preimage for some (and only one!) element in $S^{=1}$ or not a preimage for any element. Our goal is to partition $L$ into $\theta$ subsets $L_1, L_2, \ldots, L_{\theta}$ such that every subset (except for probably the last) has $n' = \frac{\alpha n}{\theta}$ elements that are preimages for some elements from $S^{=1}$. This can be done by a simple greedy algorithm, just choose $L_1$ as first $p_1$ elements from $L$ such that there are $n'$ base elements and make $p_1$ as small as possible. After that choose $L_2$ as elements from $p_1 + 1$ to $p_2$ containing $n'$ base elements and so on. Finally, $L_{\theta}$ will contain elements from $p_{\theta - 1}$ to $|L|$. After this transformation, each element $v \in S^{=1}$ has index of subset of $L$ of its base, call it $b(v)$.

    Finally, we call the tuple $(p_1, p_2, \ldots, p_{\theta - 1}, b(S^{=1}[1]), b(S^{=1}[2]), \ldots, b(S^{=1}[|S^{=1}|]))$ the trace of mapping $\phi$. It's clear that every $\phi$ has a unique trace, so it's enough to find number of mappings for each trace and just sum them.  

    For $K \subseteq L$ and $f \in \{0, 1\}$ denote $G_K(S^{=1}, S^{=0}, f)$ as a set of mappings $\phi$ from $K$ to $R$ such that
    \begin{itemize}
        \item for every $u \in K$ $(u, \phi(u)) \in E$;
        \item for every $v \in S^{=1}$ there is exactly one $u \in K$ such that $\phi(u) = v$;
        \item for every $v \in S^{=0}$ there are no $u \in K$ such that $\phi(u) = v$;
        \item if $f = 1$, then $K$ is non-empty and for vertex $w$ with the maximum number in $K$ holds $\phi(w) \in S^{=1}$;
        \item if $f = 0$, then there are no additional restrictions.
    \end{itemize}

    For every trace, consider $W(k)$ to be a set of $v \in R$ such that $b(v) = k$. In this case, the number of mappings with a given trace is \[\left(\prod_{k = 1}^{\theta - 1} |G_{L_k}(W(k), S^{=0}_i \cup (S^{=1} \setminus W(k)), 1)|\right) \cdot |G_{L_\theta}(W(\theta), S^{=0}_i \cup (S^{=1} \setminus W(\theta)), 0)|,\] so our polynomial is just a sum of these monomials for all possible traces.

\begin{lemma}
    Each variable $|G_{K}(S^{=1}, S^{=0}, f)|$ can be calculated in time $O^*(2^{(|S^{=1}|)})$.
\end{lemma}
\begin{proof}
    For every $0 \leq i \leq |K|$ let $K_i$ be the set of the first $i$ elements in $K$.

    For every $0 \leq i \leq |K|$, $T^{=1} \subseteq S^{=1}$ and $f \in \{0, 1\}$ we compute a dynamic programming $dp[i][T^{=1}][f]$ which is a number of mappings from $K_i$ to $R$ such that
    \begin{itemize}
        \item for every $u \in K_i$ $(u, \phi(u)) \in E$;
        \item for every $v \in T^{=1}$ there is exactly one $u \in K_i$ such that $\phi(u) = v$;
        \item for every $v \in S^{=0}$ there are no $u \in K_i$ such that $\phi(u) = v$;
        \item if $f = 1$, then $i > 0$ and for a single vertex $\{u\} \in K_i \setminus K_{i - 1}$ $\phi(u) \in S^{=1}$;
        \item if $f = 0$, then there are no additional restrictions.
    \end{itemize}

    It's clear that $dp[0][T^{=1}]$[f] equals to $1$ if $T^{=1} = \emptyset$ and $f = 0$ and equals $0$ otherwise. By definition, $dp[|K|][S^{=1}][f] = |G_{K}(S^{=1}, S^{=0}, f)|$.

    Consider $dp[i][T^{=1}][f]$ for some $i > 0$ and let $u$ be the single element in $K_i \setminus K_{i - 1}$. Let $N(u)$ be the set of vertices $v \in R$ such that $(u, v) \in E$.
    
    If $f = 1$, $u$ can be mapped only to some element from $T^{=1}$, so
    \[dp[i][T^{=1}][1] = \sum_{v \in N(u) \cap T^{=1}} dp[i - 1][T^{=1} \setminus v][0].\]

    If $f = 0$, $u$ can be either mapped to some element from $T^{=1}$ (the number of such mappings is calculated above) or to some element from $N(u) \cap (R \setminus S^{=1} \setminus S^{=0})$, so
    \[dp[i][T^{=1}][0] = \sum_{v \in N(u) \cap T^{=1}} dp[i - 1][T^{=1} \setminus v][0] + |N(u) \cap (R \setminus S^{=1} \setminus S^{=0})| \cdot dp[i - 1][T^{=1}][0].\]

    The number of states in dynamic programming is $2^{|S^{=1}|} \cdot n$ and each value can be calculated in $O(n)$ time, so $|G_{K}(S^{=1}, S^{=0}, f)|$ can be calculated in time $O^*(2^{(|S^{=1}|)})$.
\end{proof}

There are $\frac{n(n + 1)}{2}$ ways to choose $L_k$ for $G_{L_k}(W(k), S^{=0}_i, f)$ and $\binom{n}{\alpha n / \theta}$ ways to select a $W(k)$, so there are $\frac{n(n + 1)}{2} \binom{n}{\alpha n / \theta}$ variables $G$, each of which can be computed in time $O^*(2^{\alpha n / \theta})$. So we can compute all the variables in time $O^*(2^{\gamma n})$.

There are $\binom{n}{\theta}$ ways to select $(p_1, p_2, \ldots, p_{\theta})$ for trace and $\binom{\alpha n}{\alpha n / \theta}^{\theta}$ ways to select $(b(S^{=1}[1]), b(S^{=1}[2]), \ldots, b(S^{=1}[|S^{=1}|]))$ for trace since it is a partition of a set $S^{=1}$ of size $\alpha n$ into $\theta$ sets of size $\alpha n / \theta$, so there are $\binom{n}{\theta} \binom{\alpha n}{\alpha n / \theta}^{\theta} = O^*(2^{\theta n})$ traces and $O^*(2^{\theta n})$ monomials. So, the final polynomial has degree $\theta$ and can be constructed in time $O^*(2^{\theta n})$.
\end{proof}

\end{itemize}
   
\end{proof}

\section*{Acknowledgements} 
We~are grateful to~Amir Abboud for suggesting 
to~state explicitly the POSETH conjecture as~well~as
for Marvin Künnemann and Michael Saks for asking
whether polynomial formulations could help 
to~show hardness of~proving SETH hardness of~problems
from~\P{}. We~thank Alexander Golovnev for 
providing~us with a~useful feedback on~an~early
draft.

\bibliographystyle{alpha} 
\bibliography{main.bib}

\end{document}